\pdfoutput=1
\RequirePackage{ifpdf}
\ifpdf 
\documentclass[pdftex]{sigma}
\else
\documentclass{sigma}
\fi

\usepackage{pdfsync}
\usepackage{enumitem}
\usepackage{dsfont}
\usepackage{mathrsfs}
\catcode`\@=11
\usepackage{empheq}
\usepackage{slashed}
\usepackage{nicefrac,xfrac}

\numberwithin{equation}{section}

\newtheorem{Theorem}{Theorem}[section]
\newtheorem*{Theorem*}{Theorem}

\newtheorem{Lemma}[Theorem]{Lemma}

 { \theoremstyle{definition}

 }

\begin{document}
\allowdisplaybreaks

\newcommand{\arXivNumber}{2303.11461}

\renewcommand{\PaperNumber}{086}

\FirstPageHeading

\ShortArticleName{Unitarity of the SoV Transform for $\mathrm{SL}(2,\mathbb C)$ Spin Chains}

\ArticleName{Unitarity of the SoV Transform\\ for $\boldsymbol{\mathrm{SL}(2,\mathbb C)}$ Spin Chains}

\Author{Alexander N.~MANASHOV}

\AuthorNameForHeading{A.N.~Manashov}

\Address{Max-Planck-Institut f\"ur Physik, Werner-Heisenberg-Institut, 80805 M\"unchen, Germany}
\Email{\href{mailto:alexander.manashov@desy.de}{alexander.manashov@desy.de}}

\ArticleDates{Received March 30, 2023, in final form October 20, 2023; Published online November 04, 2023}

\Abstract{We prove the unitarity of the separation of variables transform for $\mathrm{SL}(2,\mathbb C)$ spin chains by a method based on the use of Gustafson integrals.}

\Keywords{spin chains; separation of variables; Gustafson's integrals}

\Classification{33C70; 81R12}

\section{Introduction}
\label{sect:intr}

Theory of quantum integrable models is an important part of modern theoretical physics. The~solution of such models relies on the Quantum
inverse scattering method (QISM) which includes such techniques as the algebraic Bethe ansatz (ABA)~\cite{MR549615} and separation of
variables (SoV)~\cite{MR802110,MR1239668}. The ABA allows one to effectively calculate energies and eigenstates of integrable models and
to address more complicated problems such as calculating norms~\cite{MR677006}, scalar products~\cite{MR1007797} and correlation
functions~\cite{MR763763,MR1741654}. Models with infinite-dimensional Hilbert spaces without a pseudo-vacuum state, the Toda
chain~\cite{MR626693} being the most famous example, are, however, beyond ABA's grasp. The solution of such models relies on the SoV
method proposed by Sklyanin~\cite{MR802110,MR1239668}. The method consists in constructing a map between the original Hilbert space,
${\mathbb H}_{\rm org}$, in which the model is formulated, and an auxiliary Hilbert space, ${\mathbb H}_{\rm SoV}$. This map is constructed
in such a way that a multidimensional spectral problem associated with the original Hamiltonian is reduced to a one-dimensional problem on
an auxiliary Hilbert space which usually takes the form of the Baxter $T$-$Q$ relation. Technically constructing the SoV representation is
equivalent to finding the eigenfunctions of an element of the monodromy matrix associated with the model. For the Toda chain it was done by
Kharchev and Lebedev \cite{MR1751619,Kharchev:2000yj}. Later, a regular method for obtaining eigenfunctions for models with an $R$-matrix of
the rank one\footnote{In recent years, significant progress has been made in constructing SoV representations for higher rank
finite-dimensional models, see~\cite{Cavaglia:2019pow,Derkachov:2018ewi,GromovRyan20,GromovSizov17,Gromov:2022waj,MailletNiccoli18,MR3983970,MailletNiccoli19,Maillet:2020ykb,Ryan:2021duf,RyanVolin19,
Ryan:2020rfk,Valinevich:2016cwq}.} was developed in~\cite{Derkachov:2001yn}, and at
present the SoV representation is known for a number of
models~\cite{BytskoTeschner06,Derkachov:2003qb,Derkachov:2002tf,Derkachov:2014gya,Silantyev}.

In order to be sure that the spectral problems in the original and auxiliary Hilbert spaces are equivalent, it is necessary to show that the
corresponding map, ${\mathbb H}_{\rm SoV}\mapsto {\mathbb H}_{\rm org}$, is unitary (or that the eigenfunctions form a complete set in
${\mathbb H}_{\rm org}$). If $\dim {\mathbb H}_{\rm org} < \infty$ the problem can be solved, at least in principle, by counting the
dimensions of the Hilbert spaces. For the models with infinite-dimensional Hilbert space, such as the Toda chain, the noncompact
$\mathrm{SL}(2,\mathbb C)$ spin chain, etc., the task becomes more difficult. For the Toda chain, unitarity was first established by using
harmonic analysis of Lie groups techniques~\cite{Semenov-Tian-Shansky1994,Wallach92}. However, this method is quite sophisticated and can
hardly be generalized to more complicated cases. The rigorous proof of the unitarity of the SoV transform for the Toda chain based on the
use of natural objects for the QISM was given by Kozlowski~\cite{Kozlowski15}. This technique was later applied to the modular $XXZ$
magnet~\cite{Derkachov:2018lyz}. Later it was realized~\cite{Derkachov:2016dhc} that there exists a close relation between
$\mathrm{SL}(2,\mathbb R)$ symmetric spin chains and the multidimensional Mellin--Barnes integrals studied by
Gustafson~\cite{Gustafson92,Gustafson94} that allowed to greatly simplify the proof of the unitarity of the SoV transform for
$\mathrm{SL}(2,\mathbb R)$ symmetric spin chains~\cite{Derkachov:2021wja}.

In the present paper, we apply this technique to the analysis of the noncompact spin chains with the $\mathrm{SL}(2,\mathbb C)$ symmetry
group. Such models appear in the studies of the Regge limit of scattering amplitudes in gauge theories, in QCD in
particular~\cite{Bartels:2011nz,Faddeev:1994zg,Lipatov:1993qn,Lipatov:1993yb,Lipatov:2009nt}, see
also~\cite{Derkachov:2018rot,Derkachov:2019tzo,Derkachov:2021rrf,Derkachov:2020zvv} for recent developments. The SoV representation for the
$\mathrm{SL}(2,\mathbb C)$ spin chains\footnote{To the best of our knowledge, the completeness of this representation has not yet been
addressed.}
 was constructed in~\cite{Derkachov:2001yn} while the generalization of Gustafson integrals relevant
for the $\mathrm{SL}(2,\mathbb C)$ spin chains was obtained recently in~\cite{Derkachov:2019ynh,Sarkissian:2023rwb}. Based on these
results, we present below a proof of unitarity of the SoV transform for a generic $\mathrm{SL}(2,\mathbb C)$ spin chain.

The paper is organized as follows. In Section~\ref{preliminaries}, we recall elements of the QISM relevant for further analysis. The
eigenfunctions of the elements of the monodromy matrix are constructed in Section~\ref{sect:eigenfunctions}. In Section~\ref{sect:etc}, we
calculate several scalar products of the eigenfunctions and discuss their properties. Section~\ref{sect:SoV} contains the proof of unitarity
of the SoV transform. Section~\ref{sect:summary} is reserved for a summary and several appendices contain a discussion of technical details.

\section[SL(2,C) spin chains]{$\boldsymbol{\mathrm{SL}(2,\mathbb C)}$ spin chains}\label{preliminaries}

Spin chains are quantum mechanical systems whose dynamical variables are spin generators. We consider models with
spin generators belonging to the unitary continuous principal series representation, $\mathrm T^{(s_k,\bar s_k)}$, of the unimodular group
of complex two by two matrices. Namely, each site of the chain is equipped with two sets of generators, holomorphic ($S^\alpha$) and
anti-holomorphic ones~$\big(\bar S^\alpha\big)$,
\begin{alignat*}{4}
&S^-_k=-\partial_{z_k}, \qquad&& S^0_k=z_k\partial_{z_k} + s_k, \qquad&& S^+_k=z_k^2\partial_{z_k} + 2 s_k z_k,&\\
&\bar S^-_k=-\partial_{\bar z_k},\qquad&&\bar S^0_k=\bar z_k \partial_{\bar z_k} + \bar s_k, \qquad&& \bar S^+_k=\bar z_k^2
 \partial_{\bar z_k} + 2\bar s_k \bar z_k.&
\end{alignat*}
The generators $S^\alpha_k\big(\bar S^\alpha_k\big)$ satisfy the standard $\mathfrak{sl}(2)$ commutation relations, while the generators at different
sites and holomorphic and anti-holomorphic generators commute, $\big[S^\alpha_k, \bar S^{\alpha'}_k\big]=0$. The parameters $s_k$, $\bar s_k$
specifying the representation take the form~\cite{MR0207913}
\begin{align*}
s_k=\frac{1+n_k}{2}+{\rm i}\rho_k, \qquad \bar s_k=\frac{1- n_k}{2}+{\rm i}\rho_k,
\end{align*}
where $n_k$ is an integer or half-integer number and $\rho_k$ is real, so that
\begin{align*}
s_k+\bar s_k^*=1 \qquad \text{and}\qquad s_k-\bar s_k=n_k\in \mathbb Z/2.
\end{align*}
The later condition comes from the requirement for the finite group transformations to be well defined while the former one guarantees the
unitary character of transformations and anti-hermiticity of
the generators, \smash{$\big(S_k^\alpha\big)^\dagger=-\bar S_k^\alpha$}.

The Hilbert space of the model is given by the direct product of the Hilbert spaces at each node. For a chain of length $N$, ${\mathbb
H}_N=\bigotimes_{k=1}^N \mathcal H_k$, where $\mathcal H_k = L_2(\mathbb C)$.

In the QISM~\cite{MR671263,MR1239668,MR549615,MR562799}, the dynamics of the model is determined by a family of mutually commuting
operators.
 Namely, one defines the so-called $L$-operators,
\begin{align*}
L_k(u)= u + {\rm i}\begin{pmatrix}
S_k^0 & S_k^-\vspace{1mm}\\
S_k^+ & -S_0^-
\end{pmatrix},
\qquad
\bar L_k(\bar u)= \bar u + {\rm i} \begin{pmatrix}
\bar S_k^0 & \bar S_k^-\vspace{1mm}\\
\bar S_k^+ & -\bar S_0^-
\end{pmatrix},
\end{align*}
which are the basic building blocks in the QISM. The complex variables $u$, $\bar u$ are called spectral parameters. The next important
object -- a monodromy matrix -- is given by the product of $L$ operators
\begin{gather}
T_N(u) = L_1(u+\xi_1)L_2(u+\xi_2)\cdots L_N(u+\xi_N), \notag\\
\bar T_N(\bar u) = \bar L_1\big(\bar u+\bar \xi_1\big)\bar L_2\big(\bar u+\bar \xi_2\big)\cdots \bar L_N\big(\bar u+\bar \xi_N\big),\label{monodromymatrix}
\end{gather}
where $\xi_k$, $\bar \xi_k$ are the so-called impurity parameters.\footnote{As it can already be noticed any formula in the holomorphic sector has its exact copy in the anti-holomorphic one. Therefore, from now on we write explicitly only holomorphic formulae tacitly implying its anti-holomorphic counterparts.} The entries of the monodromy matrix,
\begin{align*}
T_N(u)=\begin{pmatrix}
A_N(u) & B_N(u)\\
C_N(u) & D_N(u)
\end{pmatrix},
\end{align*}
are polynomials in $u$ with the operator valued coefficients, e.g.,
\begin{gather}
A_N(u)= u^N + u^{N-1}\big({\rm i} S^0+\Xi\big) + \sum_{k=2}^{N} u^{N-k} a_k,\nonumber\\
B_N(u)= u^{N-1} {\rm i} S^- + \sum_{k=2}^{N} u^{N-k} b_k,\label{ANBN}
\end{gather}
where $\Xi=\sum_{k=1}^N \xi_k$ and $S^0$, $S^-$ are the total generators,
\begin{align*}
S^\alpha = S^{\alpha}_1 + \cdots + S^\alpha_N.
\end{align*}
The entries of the monodromy matrix form commuting operator families~\cite{MR1616371,MR549615}
\begin{align*}
[A_N(u),A_N(v)]=[B_N(u),B_N(v)]=[C_N(u),C_N(v)]=[D_N(u),D_N(v)] =0.
\end{align*}
In particular, each entry commutes with the corresponding total generator, $S^\alpha$,
\begin{align*}
\big[S^0,A_N(u)\big]=\big[S^0,D_N(u)\big]=0 \qquad\text{and}\qquad[S^-,B_N(u)]=\big[S^+,C_N(u)\big]=0.
\end{align*}
The same equations hold for the anti-holomorphic operators $\bar A_N$, $\bar B_N$, $\bar C_N$, $\bar D_N$ and, of course, the holomorphic and
anti-holomorphic operators commute. Moreover it can be checked that if the impurity parameters satisfy the constraint $\bar \xi_k=\xi_k^*$
for all $k$, the following relations between holomorphic and anti-holomorphic operators hold:
\begin{align*}
(A_N(u))^\dagger = \bar A_N(u^*), \qquad (B_N(u))^\dagger = \bar B_N(u^*),
\end{align*}
etc. This ensures that the operators $a_k$ and $\bar a_k$ in the expansion of $A_N(u)$, \eqref{ANBN}, and $\bar A_N(u)$, are adjoint
to each other $a_k^\dagger =\bar a_k$ \big($b_k^\dagger =\bar b_k$ etc.\big).

The commutativity of the operators $A_N(u)$, $B_N(u)$, $C_N(u)$, $D_N(u)$ implies that the following families of self-adjoint operators:
\begin{gather*}
\mathfrak A_N=\big\{{\rm i}S^0, {\rm i}\bar S^0, a_k+\bar a_k, {\rm i}(a_k-\bar a_k), k=2,\dots,N\big\},
\\
\mathfrak B_N=\big\{{\rm i}S^-, {\rm i}\bar S^-, b_k+\bar b_k, {\rm i}(b_k-\bar b_k), k=2,\dots,N\big\},
\end{gather*}
(and similarly for others) are commutative and can be diagonalized simultaneously.\footnote{The impurity parameters must also satisfy
the condition ${\rm i}\big(\xi_k -\bar \xi_k\big)=r_k$, where $r_k$ are half-integers.} The corresponding eigenfunctions provide a convenient basis --
Sklyanin's representation of Separated Variables (SoV) -- for the analysis of spin chain models~\cite{MR1239668}.

The operators $B_N$ and $C_N$, ($A_N$ and $D_N$) are related to each other by the inversion transformation, see~\cite{Derkachov:2014gya} for detail, so it is sufficient to construct eigenfunctions for the operators~$B_N$ and $A_N$. The
eigenfunctions of $B_N$ for the homogeneous chain were constructed in~\cite{Derkachov:2001yn} and later on for the operator
$A_N$~\cite{Derkachov:2014gya}. Extending this approach to the inhomogeneous case is rather straightforward.

\section{Eigenfunctions}\label{sect:eigenfunctions}

In this section, we present explicit expressions for the eigenfunctions of the operators $B_N$ and~$A_N$ for a generic inhomogeneous spin
chain with impurities. We start with the operator $B_N$ where the construction follows the lines of~\cite{Derkachov:2001yn} with
minimal modifications.

\subsection[B\_N operator]{$\boldsymbol{B_N}$ operator}\label{subs:BN}

Let $\Lambda_n$ be an integral (layer) operator which maps functions of $n-1$ variables into functions of $n$ variables and depends on
the spectral parameters $x$, $\bar x$ and the complex vectors $\gamma$, $\bar \gamma$ of dimension $2n-2$
\begin{gather}
[\Lambda_n(x|\gamma)f](z_1,\dots,z_n)\nonumber\\
\qquad{} =\idotsint \Lambda_n(x|\gamma)(z_1,\dots,z_n|w_1,\dots,w_{n-1})
f(w_1,\dots,w_{n-1})\prod_{k=1}^{n-1} {\rm d}^2 w_k.\label{LNoperator}
\end{gather}
The kernel is given by the following expression:
\begin{gather*}
\Lambda_n(x|\gamma)(z_1,\dots,z_n|w_1,\dots,w_{n-1}) =\prod_{k=1}^{n-1} D_{\gamma_{2k-1}-{\rm i}x}(z_{k}-w_{k})
D_{\gamma_{2k}+{\rm i}x}(z_{k+1}-w_{k}),
\end{gather*}
where the function $D_\alpha(z)$ (propagator) is defined as follows:
\begin{gather}\label{Dalpha}
D_{\alpha}(z)\equiv D_{\alpha,\bar\alpha}(z,\bar z) = {z^{-\alpha} \bar z^{-\bar\alpha}}.
\end{gather}
We will assume that the indices $\alpha$, $\bar\alpha$ satisfy the condition $[\alpha] \equiv \alpha-\bar\alpha\in \mathbb Z$ so that the
propagator is a single-valued function on the complex plane. It implies that the parameters $\gamma_k$ and $x$ have the form
\begin{align}\label{gammaxform}
\gamma_k=\frac12+\frac{r_k}{2}+{\rm i}\sigma_k, \qquad \bar \gamma_k=\frac12- \frac{r_k}{2}+{\rm i}\sigma_k, \qquad
x=\frac{{\rm i}m}2 + \nu, \qquad \bar x=-\frac{{\rm i}m}2 + \nu.
\end{align}
The numbers $\{m, r_1,\dots, r_{2N-2}\}$ are either integer or half-integer and depending on this we call the corresponding variables
integer (half-integer).
 The continuous parameters $\sigma_k$ and $\nu$ are subject to the constraints
\begin{align*}
\operatorname{Im}(\sigma_{2k+1} - \nu)> -\nicefrac 12 \qquad\text{and}\qquad \operatorname{Im}(\sigma_{2k} + \nu)> - \nicefrac12,
\end{align*}
which guarantee the convergence of the integral~\eqref{LNoperator} for a smooth function $f$ with finite support.
In the case we are most interested in, $\gamma_k + \bar\gamma_k=1$, the parameters $\sigma_k\in \mathbb{R}$, and the variable $\nu$ lies
in the strip $-\nicefrac12<\operatorname{Im}\nu<\nicefrac12$.

The operators $\Lambda_n$ possess two important properties:
\begin{enumerate}\itemsep=0pt
\item[(i)]
 Let $\rho$ be a map which takes $M$-dimensional vectors
\begin{align*}
\gamma=(\gamma_1,\dots,\gamma_{M}), \qquad
\bar \gamma=(\bar \gamma_1,\dots,\bar \gamma_{M})
\end{align*}
to vectors of dimension $M-2$ as follows:
\begin{align*}
\rho \gamma=(\gamma'_2,\gamma'_3,\dots,\gamma'_{M-1}), \qquad
\rho \bar \gamma=(\bar \gamma'_2,\bar \gamma'_3,\dots,\bar \gamma'_{M-1}),
\end{align*}
where $a' \equiv 1-a$. It can be shown that the operators $\Lambda_n$ and $\Lambda_{n-1}$ obey the following exchange relation:
\begin{align}\label{Bexchange}
\Lambda_n(u|\gamma) \Lambda_{n-1}(v|\rho \gamma) & = \omega_n(\gamma, u,v) \Lambda_n(v|\gamma) \Lambda_{n-1}(u|\rho \gamma).
\end{align}
Here $\gamma (\bar\gamma)$ is $(2n-2)$-dimensional vector and the factor $\omega_n$ is given by the following expression:
\begin{align}\label{omegan}
\omega_n(\gamma, u,v) & = \prod_{m=1}^{n-1}
\boldsymbol \Gamma\left[
\frac{\gamma_{2m-1}-{\rm i}v,\bar\gamma_{2m}+{\rm i}\bar v}
{\gamma_{2m-1}-{\rm i}u,\bar\gamma_{2m}+{\rm i}\bar u}\right]
=
\prod_{m=1}^{n-1} \boldsymbol \Gamma\left[
\frac{\bar \gamma_{2m-1}-{\rm i}\bar v,\gamma_{2m}+{\rm i} v}
{\bar \gamma_{2m-1}-{\rm i}\bar u,\gamma_{2m}+{\rm i} u}\right],
\end{align}
where
\begin{align*}
\boldsymbol \Gamma\left[\frac{a_1,a_2,\dots,a_n}{b_1,b_2,\dots,b_m}\right]\equiv
\frac{\prod_{k=1}^n\boldsymbol\Gamma[a_k]}
{\boldsymbol\prod_{k=1}^m\boldsymbol\Gamma[b_k]}
\end{align*}
and $\boldsymbol \Gamma$ is the Gamma function of the complex field $\mathbb C$~\cite{MR2125927}
\begin{align*}
\boldsymbol\Gamma[u] \equiv \boldsymbol\Gamma[u,\bar u] =
{\Gamma(u)}/{\Gamma(1-\bar u)}.
\end{align*}
The relation~\eqref{Bexchange} is a direct consequence of the exchange relation for the propagators, see~\eqref{exchange-rel}. Its proof
is exactly the same as for the homogeneous spin chain. For more details, see~\cite{Derkachov:2001yn,Derkachov:2014gya}.

\item[(ii)] Let us choose the vector $\gamma$ as follows
\begin{align}
\gamma =(s_1-{\rm i}\xi_1,s_2+{\rm i}\xi_2,s_2-{\rm i}\xi_2,\dots,s_{N-1} + {\rm i}\xi_{N-1},s_{N-1} - {\rm i}\xi_{N-1},s_N + {\rm i}\xi_N),
\notag\\
\bar\gamma =(\bar s_1-{\rm i}\bar\xi_1,\bar s_2+{\rm i}\bar\xi_2,\bar s_2-{\rm i}\bar\xi_2,\dots,\bar s_{N-1} + {\rm i}\bar\xi_{N-1},
\bar s_{N-1} - {\rm i}\bar\xi_{N-1},\bar s_N + {\rm i}\bar \xi_N),\label{gammaN}
\end{align}
where $s_k$ and $\xi_k$ are the spins and impurity parameters of the spin chain, respectively. For such a choice of the vector $\gamma$,
the operator $B_N(x)$ annihilates $\Lambda_N(x|\gamma)$~\cite{Derkachov:1999pz,Derkachov:2001yn}
\begin{align}\label{Beq}
B_N(x) \Lambda_N(x|\gamma)=0.
\end{align}
\end{enumerate}

Let us define a function
\begin{align*}
\Psi^{(N)}_{p,x}(z)
& \equiv \Psi^{(N)}_{p,x_1,\dots,x_{N-1}} (z_1,\dots, z_N)\\
& = \pi^{-N^2/2} |p|^{N-1} \int {\rm d}^2 z U_{x_1,\dots,x_{N-1}}(z_1,\dots,z_N| z)
{\rm e}^{{\rm i} (p z +\bar p \bar z)},
\end{align*}
where the kernel $ U_{x_1, \dots, x_{N-1} }$ is given by the product of the layer operators,
\[
U_{x_1,\dots,x_{N-1}} =
 \varpi (x|\gamma)
\Lambda_N (x_1|\gamma)
 \Lambda_{N-1}(x_2|\rho\gamma)
 \Lambda_{N-2}\big(x_3|\rho^2\gamma\big)
 \cdots\Lambda_2\big(x_{N-1}|\rho^{N-2}\gamma\big),
\]
and $\gamma $ is given by~\eqref{gammaN}.
Equation~\eqref{Bexchange} guarantees that $U_{x_1,\dots,x_{N-1}}\sim U_{x_{i_1},\dots,x_{i_{N-1}}}$ for any permutation of
$x_1,\dots, x_{N-1}$. The kernel $U_x$ becomes totally symmetric for the following choice of the prefactor $\varpi (x|\gamma)$:
\begin{align}\label{varphifactor}
\varpi({x}|\gamma) &=\varpi(x_1,\dots,x_{N-1}|\gamma)=
\prod_{m=1}^{N-1} \prod_{k=1}^{m}\varpi_1\big(x_k|\rho^{m-1}\gamma\big),
\end{align}
where
\begin{align*}
\varpi_1(x|\gamma)=\varpi_1(x|\gamma_1,\dots,\gamma_{2n}) =
 \prod_{m=1}^{n} \boldsymbol\Gamma[\gamma_{2m-1}-{\rm i}x,\bar\gamma_{2m} + {\rm i}\bar x].
\end{align*}
Thus the function $ \Psi^{(N)}_{p,x_1,\dots,x_{N-1}} $ is a symmetric function of the variables $x_1,\dots,x_{N-1}$.
 Together with \eqref{Beq} it implies that
\begin{align*}
B_N(x_k)\Psi^{(N)}_{p,x_1,\dots,x_{N-1}} =0 \qquad \text{for}\quad k=1,\dots, N-1.
\end{align*}
Invariance of the kernel $U_{x_1m\dots,x_{N-1}}(z_1,\dots, z_N|z)$ under shifts
\begin{align*}
U_{x_1\dots,x_{N-1}}(z_1+w,\dots, z_N+w|z+w) =U_{x_1\dots,x_{N-1}}(z_1,\dots, z_N|z)
\end{align*}
results in
\begin{align}\label{Seq}
{\rm i}S^{-} \Psi^{(N)}_{p,x_1,\dots,x_{N-1}} = p \Psi^{(N)}_{p,x_1,\dots,x_{N-1}}, \qquad
{\rm i}\bar S^{-} \Psi^{(N)}_{p,x_1,\dots,x_{N-1}} = \bar p \Psi^{(N)}_{p,x_1,\dots,x_{N-1}}.
\end{align}
It follows then from equations~\eqref{ANBN}, \eqref{Beq} and \eqref{Seq} that\footnote{We recall that the variables $x_k$, $\bar x_k$,
$k=1,\dots, N-1$ take the form
$x_k={{\rm i}n_k}/{2}+\nu_k$, $ \bar x_k=-{{\rm i}n_k}/{2}+\nu_k$,
where, depending on the spin and impurities parameters, all $n_k$ are either integer or half-integer numbers. }
\begin{align*}
B_N(u)\Psi^{(N)}_{p,x}(z)=p\prod_{k=1}^{N-1}(u-x_k)\Psi^{(N)}_{p,x}(z), \qquad
\bar B_N(\bar u)\Psi^{(N)}_{p,x}(z)= \bar p\prod_{k=1}^{N-1}(\bar u-\bar x_k)\Psi^{(N)}_{p,x}(z).
\end{align*}

For $N=1$, the functions $\Psi_{p}^{(1)}(z,\bar z) = \pi^{-1/2} {\rm e}^{{\rm i}(p z + \bar p\bar z)}$ form the complete orthonormal system in~${\mathbb H}_1=L_2(\mathbb C)$. The aim of this paper is to extend this statement to $N>1$. Namely, we will show in Section~\ref{sect:SoV}
that if the spins and impurities parameters of the spin chain obey the ``unitarity'' condition,
\begin{align}\label{gammaunitarity}
\gamma_k+\bar\gamma_k^*=1,
\end{align}
for all $k$ ($\gamma_k$ has the form~\eqref{gammaxform} with $\sigma_k\in\mathbb R$ ) then the set of functions $\big\{\Psi^{(N)}_{p,x},
x_k=\bar x_k^* (\nu_k\in \mathbb R),\allowbreak k=1,\dots, N-1\big\}$ is complete in ${\mathbb H}_N=(\bigotimes L_2(\mathbb C))^N$.

Note that the functions $\Psi_{p,x}^{(N)}$ are well defined for the complex parameters $\nu_k$ in the vicinity of the real line. For
further analysis, it will be useful to consider regularized functions, \smash{$\Psi_{p,x}^{(N),\epsilon}$}, by relaxing the last of the
conditions~\eqref{gammaunitarity} to $\gamma_{2N-2}+\bar\gamma_{2N-2}^*=1+2\epsilon$. This can be achieved by shifting the impurity
parameter $\xi_N\to\xi_N - {\rm i}\epsilon$,\footnote{Of course, one also can regularize the function by shifting the parameter $\gamma_1$
instead of $\gamma_{2N-2}$, $\gamma_{1}+\bar\gamma_{1}^*=1+2\epsilon$.} i.e.,
\begin{align}\label{Psiepsilon}
\Psi_{p,x}^{(N),\epsilon}(z) \overset{\text{def}}{=} \Psi_{p,x}^{(N)}(z)\Big|_{\xi_N\to\xi_N - {\rm i}\epsilon}.
\end{align}

\subsection[A\_N operator]{$\boldsymbol{A_N}$ operator}\label{subs:AN}

Construction of the eigenfunctions of the operator $A_N$ follows the scheme described in the previous subsection. We define a layer
operator $ \Lambda^\prime_n$ which maps functions of $n-1$ variables into functions of $n$ variables
\begin{gather*}
[\Lambda^\prime_n(x|\gamma)f](z_1,\dots,z_n) \\
\qquad = \idotsint \Lambda^\prime_n(x|\gamma)(z_1,\dots,z_n|w_1,\dots,w_{n-1}) f(w_1,\dots,w_{n-1})\prod_{k=1}^{n-1} {\rm d}^2 w_k,
\end{gather*}
where the kernel is given by the following expression:
\begin{gather*}
\Lambda^\prime_n(x|\gamma)(z_1,\dots,z_n|w_1,\dots,w_{n-1})\\
\qquad{}= D_{\gamma_{2n-1}-{\rm i}x}(z_{n}) \prod_{k=1}^{n-1} D_{\gamma_{2k-1}-{\rm i}x}(z_{k}-w_{k}) D_{\gamma_{2k}+{\rm i}x}(z_{k+1}-w_{k}).
\end{gather*}
The layer operator $ \Lambda^\prime_n$ depends on the spectral parameters $x(\bar x)$ and the vector $\gamma(\bar \gamma)$ of dimension~$2n-1$ which have the form~\eqref{gammaxform}.
These operators satisfy the exchange relation{\samepage
\begin{align*}
\Lambda^\prime_n(u|\gamma)\Lambda^\prime_{n-1}(v|\rho \gamma) & = \omega_n(\gamma, u,v)
\Lambda^\prime_n(v|\gamma)\Lambda^\prime_{n-1}(u|\rho \gamma),
\end{align*}
and the factor $\omega_n$ is defined in~\eqref{omegan}.}

Let $\Phi^{(N)}_{x}(z)$ be the following function:
\begin{gather*}
\Phi^{(N)}_{x}(z) \equiv \Phi^{(N)}_{x_1,\dots,x_{N}}(z_1,\dots, z_N)\\
\phantom{\Phi^{(N)}_{x}(z) }{} = \pi^{-N^2/2}
\varpi (x|\gamma)\big[
\Lambda^\prime_N (x_1|\gamma)
 \Lambda^\prime_{N-1}(x_2|\rho\gamma)
 \dots
 \Lambda^\prime_1\big(x_{N}|\rho^{N-1}\gamma\big)\big](z_1,\dots,z_N),
\end{gather*}
where $\gamma$ is $(2N-1)$-dimensional vector and the prefactor $\varpi$ is given by equation~\eqref{varphifactor}. For such a~choice of
$\varpi$ the function $\Phi^{(N)}_{x}$ is a symmetric function of the variables $x_1,\dots,x_N$.

It can be shown that the operator $A_N(x)$ annihilates the layer operator $\Lambda^\prime_N(x|\gamma)$,
\begin{align*}
A_N(x) \Lambda^\prime_N(x|\gamma)=0,
\end{align*}
for the following choice of the vector $\gamma$:
\begin{align*}
\gamma =(s_1-{\rm i}\xi_1,s_2+{\rm i}\xi_2,s_2-{\rm i}\xi_2,\dots,
s_N + {\rm i}\xi_N, s_N - {\rm i}\xi_N),\\
\bar\gamma =(\bar s_1-{\rm i}\bar\xi_1,\bar s_2+{\rm i}\bar\xi_2,\bar s_2-{\rm i}\bar\xi_2,\dots,
\bar s_N + {\rm i}\bar \xi_N,\bar s_N - {\rm i}\bar \xi_N).
\end{align*}
Taking into account polynomiality of $A_N(u)$, see equation \eqref{ANBN}, one obtains
\begin{align*}
A_N(u)\Phi^{(N)}_{x}(z)=\prod_{k=1}^{N}(u-x_k)\Phi^{(N)}_{x}(z), \qquad
\bar A_N(\bar u)\Phi^{(N)}_{x}(z)= \prod_{k=1}^{N}(\bar u-\bar x_k)\Phi^{(N)}_{x}(z).
\end{align*}
Again, the variables $x_k$, $\bar x_k$ are integers (half-integers) for all $k$. We will show that these functions, $\{\Phi^{(N)}_{x}(z), \,
x_k=\bar x_k^*, \, k=1,\dots, N \}$, form a complete set in the Hilbert space ${\mathbb H}_N$.

\section{Scalar products, momentum representation, etc.}\label{sect:etc}

\begin{figure}[t]\centering
\includegraphics[width=0.72\linewidth]{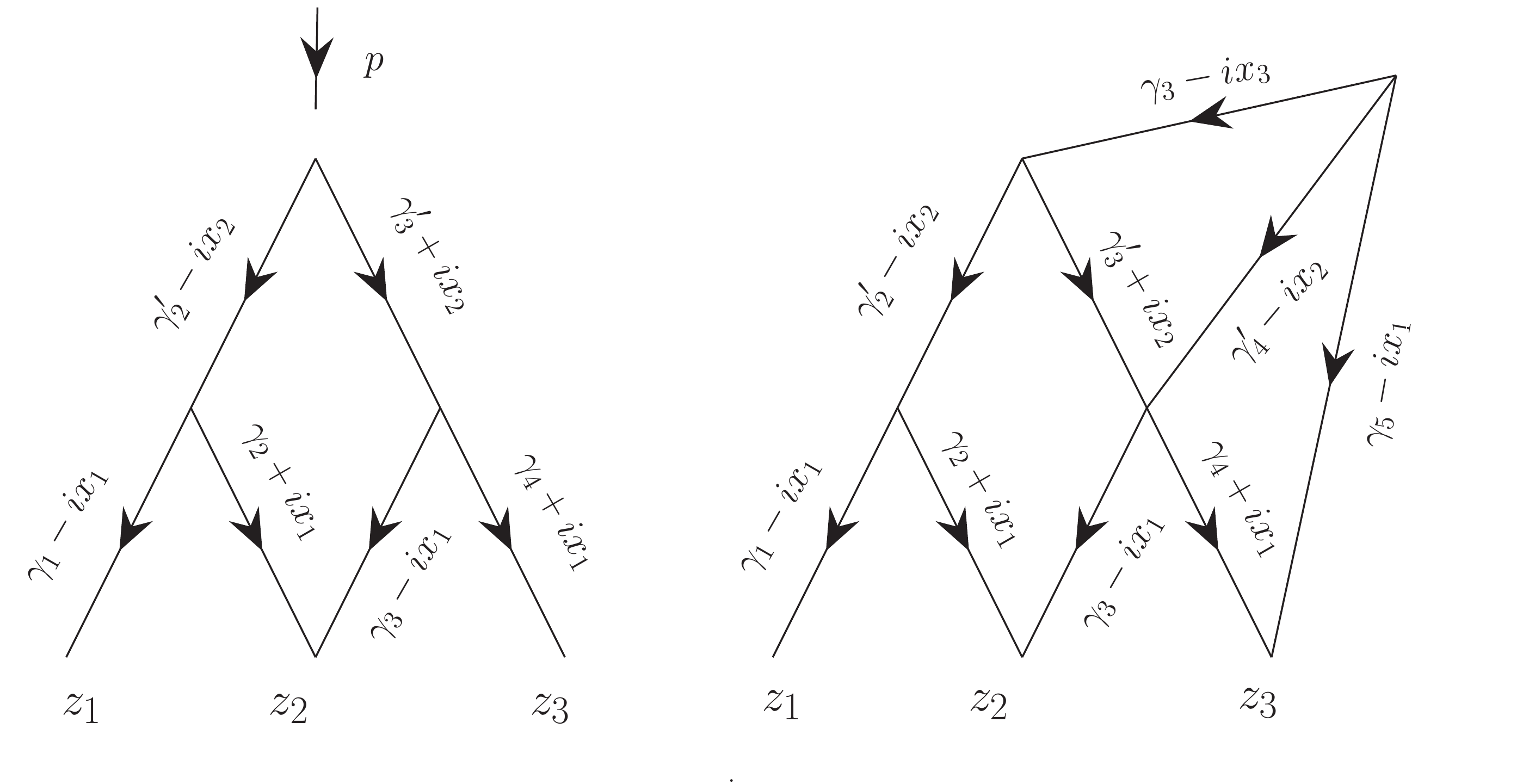}
\caption{The diagrammatic representation for the function $\Psi$ (left) and $\Phi$ (right) for $N=3$.
The arrow
from $z$ to $w$ with an index $\alpha$ stands for the propagator $D_{\alpha}(z-w)$, equation~\eqref{Dalpha}.}\label{diag:examples}
\end{figure}

The functions constructed in the previous section are given by multidimensional integrals. In this section, we show that these integrals
converge for the parameters $\nu_k$ in the vicinity of real axis. To this end, it will be quite helpful, as was advocated in~\cite{Derkachov:2001yn}, to visualize the integrals as Feynman diagrams. The examples for $N=3$ are shown in
Figure~\ref{diag:examples}.
It will be convenient to convert diagrams (functions) to momentum space 
\begin{align*}
\Psi(z_1,\dots, z_N)={\pi^{-N}}\idotsint \widetilde \Psi(p_1,\dots,p_N) {\rm e}^{{\rm i} \sum_{k=1}^N (p_kz_k + \bar p_k \bar z_k)} {\rm d}^2p_1\cdots {\rm d}^2p_N.
\end{align*}
In momentum space the function $\Psi^{(N),\epsilon}_{p, x}$, equation~\eqref{Psiepsilon}, takes the form
\begin{align*}
\widetilde \Psi^{(N),\epsilon}_{p, x}(p_1,\dots, p_N) = \delta^{(2)}\left(p-\sum_{k=1}^Np_k\right)
\Psi^{(N),\epsilon}_{x}(p_1,\dots, p_N).
\end{align*}
Let us remark here that the ``$\epsilon$'' regularization is reduced to a multiplication by the factor $(p_N \bar p_N)^\epsilon$
\begin{align}\label{Psiespilon}
\Psi^{(N),\epsilon}_{x}(p_1,\dots, p_N) & =(p_N \bar p_N)^\epsilon \Psi^{(N)}_{x}(p_1,\dots, p_N).
\end{align}
The function $\Psi^{(N),\epsilon}_{x}$ can be read from the Feynman diagram in Figure~\ref{diag:examples} as follows:
\begin{align}\label{PsiNdefinition}
\Psi^{(N),\epsilon}_{x}(p_1,\dots,p_N) =\idotsint \mathcal J_x^{\epsilon}(\{p_k\},\{\ell_{ij}\}) \prod_{1\leq j\leq i\leq N-2} {\rm d}^2\ell_{ij},
\end{align}
with the integrand $\mathcal J_x^{\epsilon}(\{p_k\},\{\ell_{ij}\})$ given by the product of the propagators, $D_\alpha(k)$. Up to a~momentum independent factor
\begin{align*}
\mathcal J_x^{\epsilon}(\{p_k\},\{\ell_{ij}\})\simeq \prod_{k=1}^{N-1}\prod_{j=1}^k D_{\alpha_{kj}}(\ell_{k,j}-\ell_{k-1,j-1})
D_{\beta_{kj}}(\ell_{k-1,j}-\ell_{k,j}),
\end{align*}
where $\ell_{k0}\equiv0$, $\ell_{k-1,k}\equiv p$ and $\ell_{N-1,j}=(p_1+ \dots + p_j)$. The indices $\alpha_{kj}$, $\beta_{kj}$ take the
following values:
\begin{align*}
\alpha_{kj} =\gamma_{2j-1}^{(N-k)}+{\rm i}x_{N-k}, \qquad \beta_{kj} =\gamma_{2j}^{(N-k)}-{\rm i}x_{N-k},
\end{align*}
where we introduced the notations:
\begin{align*}
 a^{(1)}=a^\prime =1-a \qquad\text{ and } \qquad a^{(k+1)}=1-a^{(k)}.
\end{align*}
In many cases, Feynman diagrams can be evaluated diagrammatically. In particular, the computation of diagrams for the scalar product of
$\Psi$ ($\Phi$) functions is based on the successive application of the exchange relation~\eqref{exchange-rel} to the diagram.

Let us consider the scalar product of two functions $\Psi^{(N),\epsilon}_{p, x}$ and $\Psi^{(N),\epsilon'}_{q, y}$
\begin{align}\label{BBproduct}
\Big(\Psi^{(N),\epsilon'}_{q, y},\Psi^{(N),\epsilon}_{p, x}\Big) & = \pi \delta^2(p-q) (p\bar p)^{\epsilon+\epsilon'}
I^{\epsilon,\epsilon'}(x,y),
\end{align}
where
\begin{align}\label{Ie}
I^{\epsilon,\epsilon'}(x,y) & = \frac1\pi (p\bar p)^{-\epsilon -\epsilon'}\idotsint \delta^{(2)}\bigg(p-\sum_k p_k\bigg)
 \Psi^{(N),\epsilon}_{x}(\vec{p}) \left(\Psi^{(N),\epsilon'}_{y}(\vec{p}) \right)^\dagger \prod_{j=1}^N {\rm d}^2p_j.
\end{align}
The function $I^{\epsilon,\epsilon'}_p(x,y)$ is given by the Feynman diagram shown in Figure~\ref{diag:scalarproducts} in Appendix~\ref{sect:Diagram}
 (left panel), which
is a multidimensional integral
\begin{align}\label{intI}
I^{\epsilon,\epsilon'}_p(x,y) = \idotsint\mathcal I^{\epsilon\epsilon'}_{x,y}(p,\{\ell_{pr}\}|\gamma)\prod_{p,r=1}^{N-1} {\rm d}^2\ell_{pr}
\end{align}
with the integrand given by the product of the propagators. The diagram can be evaluated in a~closed form by successively applying the
exchange relation~\eqref{exchange-rel}, that is equivalent to calculating the loop integrals in a certain order. The answer takes the form
\begin{align}
I^{\epsilon,\epsilon'}(x,y) & = \mathscr C_N(\gamma)
 \boldsymbol\Gamma\left[\frac{\epsilon+\epsilon' + {\rm i} X -{\rm i}\bar Y^*}
 {\epsilon+\epsilon'}\right]
 \frac
 {
 \prod_{k,j=1}^{N-1} \boldsymbol\Gamma[{\rm i}(y_k^*-\bar x_j)]}{
 \prod_{k=1}^{N-1} \bar \phi_N(\bar x_k) (\phi_N(y_k))^*}
\notag\\
&=\mathscr C_N(\gamma)
 \boldsymbol\Gamma\left[
 \frac{\epsilon+\epsilon' + {\rm i}\bar X -{\rm i}Y^*}{\epsilon+\epsilon'}
 \right]
 \frac{
 \prod_{k,j=1}^{N-1} \boldsymbol\Gamma[{\rm i}(\bar y_k^* - x_j)]}
 {\prod_{k=1}^{N-1} \phi_N( x_k) (\bar \phi_N(\bar y_k))^*},\label{twoform}
\end{align}
where $X=\sum_{k=1}^{N-1} x_k$, $Y=\sum_{k=1}^{N-1} Y_k$ and
\begin{gather*}
 \phi_N( x) =\boldsymbol\Gamma\bigl[\gamma_{2N-3} -{\rm i}x,\gamma^{(1)}_{2N-4} -{\rm i}x, \gamma_{2N-5} -{\rm i}x,\dots, \gamma_{N}^{(N-3)}-{\rm i} x\bigr],\\
\bar \phi_N(\bar x) =\boldsymbol\Gamma\bigl[
\bar \gamma_{2N-3} -{\rm i}\bar x, \bar \gamma^{(1)}_{2N-4} -{\rm i}\bar x, \bar \gamma_{2N-5} -{\rm i}\bar x,\dots,
\bar\gamma_{N}^{(N-3)}-{\rm i}\bar x\bigr].
\end{gather*}
For the sign factor $\mathscr C_N(\gamma)$, we get
\begin{align}\label{BBCN}
\mathscr C_{N}(\gamma_1,\gamma_2,\dots,\gamma_{2N-2})
=
\begin{cases}
1, & \text{odd } N,
\\
(-1)^{\sum_{k=1}^{N-3}\bigl[\gamma_{2N-2-k}^{(k-1)}-\gamma_{N}^{(N-3)}\bigr]},
& \text{even } N.
\end{cases}
\end{align}
Here $[a]\equiv a-\bar a$. Details of the calculation can be found in Appendix~\ref{app:scalarproduct}.

Let us show now that integrations in~\eqref{intI} can be done in an arbitrary order. The integrand in~\eqref{intI},
$\mathcal I^{\epsilon\epsilon'}_{x,y}(p,\{\ell_{pr}\}|\gamma)$, is given by
the product of the propagators $D_{\alpha}(k)$, with each index being of the form $\alpha= \frac12 +\frac{n}2 + {\rm i}\sigma$, momentum $k$
being a linear combination of loop momenta,~$\ell_{ij}$, and the external momentum $p$. Since
\begin{align*}
\big | D_\alpha(k)\big | =\big | k^{-\alpha} \bar k^{-\bar \alpha}\big | = |k|^{-1 +2 \operatorname{Im}\sigma}=D_{\nicefrac12-\operatorname{Im}\sigma}(k)
\end{align*}
then for the parameters $\gamma$ satisfying the unitarity condition~\eqref{gammaunitarity}, and $x_k$, $y_k$ having the form
\begin{align}\label{xydef}
x_k={{\rm i}n_k}/{2}+\nu_k, \qquad y_k={{\rm i}m_k}/{2}+\mu_k,
\end{align}
one obtains for the modulus of the integrand
\begin{align*}
\big|\mathcal I^{\epsilon\epsilon'}_{x,y}(p,\{\ell_{pr}\}|\gamma)\big| =
\mathcal I^{\epsilon\epsilon'}_{\underline{x},\underline{y}}(p,\{\ell_{pr}\}|\underline{\gamma})>0,
\end{align*}
where the underlined variables are: $\underline{\gamma}=(\nicefrac12,\dots,\nicefrac12)$,
\begin{align*}
(\underline{x})_k=\operatorname{Im}(\nu_k)=\epsilon_k, \qquad (\underline{y})_k=\operatorname{Im}(\mu_k)=\epsilon'_k.
\end{align*}
Thus the integral of $\vert\mathcal I^{\epsilon\epsilon'}_{x,y}(p,\{\ell_{pr}\}|\gamma)\vert$ is a particular case of the integral
 \eqref{intI} which
was calculated by performing loop integrations in a certain order. Since all integrals converge under the conditions
\begin{align*}
\epsilon_{kj}\equiv \epsilon_k+\epsilon'_j >0\qquad\text{for}\quad k,j=1,\dots, N-1\qquad\text{and} \qquad \epsilon+\epsilon'
>\sum_{k=1}^{N-1} (\epsilon_k+\epsilon'_k),
\end{align*}
by Fubini theorem, the integral~\eqref{intI} exists and the integrations can be done in an arbitrary order.

The following statements can immediately be deduced from this result:
\begin{itemize}\itemsep=0pt
\item For any bounded function $\varphi(p,x)$ with a finite support the function
\begin{align}\label{Psixepsilon}
\Psi^\epsilon_\varphi = \idotsint \varphi(p,x) \Psi^{(N),\epsilon}_{p, x^\epsilon} {\rm d}^2p \mathcal D x_1\cdots \mathcal D x_{N-1},
\end{align}
where $x^\epsilon =(x_1 + {\rm i}\epsilon_1,\dots, x_{N-1}+ {\rm i}\epsilon_{N-1})$, $x_k={\rm i}n_k/2+\nu_k$, $\epsilon_k>0$,
 $\epsilon>\sum_{k=1}^{N-1} \epsilon_k$ and
\[
 \int \mathcal D x_k \equiv \sum_{n_k=-\infty}^\infty \int_{-\infty}^{\infty} {\rm d}\nu_k,
\]
belongs to the Hilbert space ${\mathbb H}_N$, $\|\Psi^\epsilon_\varphi\|^2 <\infty$, for sufficiently small $\epsilon$.

\item It follows from the finiteness of the integral $I^{\epsilon,\epsilon'}_p(x,y)$, equation~\eqref{Ie}, that the function
 \smash{$\Psi_{x}^{(N),\epsilon}(\vec p)$}, equation~\eqref{PsiNdefinition}, exists almost for all $\vec{p}$ for the separated variables $x_k$
 close to the real axis:
 \[
 \operatorname{Im} \nu_k = \frac12\operatorname{Im}(x_k+\bar x_k) \sim 0 \qquad \text{for all $k$}
 \]
 and $\Psi_{x}^{(N),\epsilon}(\vec p)$ is a continuous function of $\nu_k$ in this region. Indeed, let us fix $m<N$ and put
 $u_m=\text{Re}\nu_m$ and $v_m=\operatorname{Im}\nu_m $, $ |v_m| <\delta$. One gets the following estimate for the
 integrand~\eqref{intI}:
\begin{align}\label{estimate1}
|\mathcal J_x^{\epsilon}(\{p_k\},\{\ell_{ij}\})|< |\mathcal J_{x_+}^{\epsilon}(\{p_k\},\{\ell_{ij}\})|
+|\mathcal J_{x_-}^{\epsilon}(\{p_k\},\{\ell_{ij}\})|,
\end{align}
where $x_\pm$ are defined as follows: for $k\neq m$ $(x_\pm)_k=x_k$ and for $k=m$ $(x_\pm)_m=u_m \pm i\delta$. The integrals of the
functions on the right-hand side of~\eqref{estimate1} are finite for sufficiently small~$\delta$. It follows then from the Lebesgue theorem that the
function $\Psi_{x}^{(N),\epsilon}(\vec p)$ is continuous in the variable $\nu_m$.\footnote{Since the integrand is analytic function of
$\nu_k$ $\Psi_{x}^{(N),\epsilon}(\vec p)$ is an analytic function of $\nu_k$ in the vicinity of the real axis.}
\end{itemize}

The scalar product of the functions $\Psi^{(N)}_{p,y}$ and $\Phi^{(N)}_{x}$ constructed in~Section~\ref{subs:AN} can be calculated in a
similar way. Note that there is no need to introduce ``$\epsilon$" regulator here. The corresponding integral is absolutely convergent when
$\operatorname{Im} (\nu_k +\mu_j)>0$ for all $k$, $j$ ($x_k$, $y_j$ given by~\eqref{xydef}).
The scalar product takes the form
\begin{align}
\big(\Psi^{(N)}_{p,y}|\Phi^{(N)}_{x} \big) ={}& \mathrm C^{AB}_N(\gamma) |p|^{N-1}
 (-{\rm i}p)^{- G_N - {\rm i}X} ({\rm i}\bar p)^{-\bar G_N - {\rm i}\bar X}\nonumber\\
&\times\frac{
\prod_{k=1}^N \prod_{j=1}^{N-1} \boldsymbol\Gamma[{\rm i}( \bar y_j^* - x_k)]}
{\left(\prod_{j=1}^{N}\vartheta_N(x_j)\right) \left(\prod_{j=1}^{N-1}\bar\vartheta_N(\bar y_j)\right)^\dagger}
,\label{ABproduct}
\end{align}
where
\begin{align*}
\vartheta_N(x)=
\prod_{k=1}^{N}
\boldsymbol \Gamma\bigl[\gamma^{(k-1)}_{2N-k}-{\rm i} x_j\bigr], \qquad \bar \vartheta_N(\bar x)=
\prod_{k=1}^{N}
\boldsymbol\Gamma\bigl[\bar \gamma^{(k-1)}_{2N-k}-{\rm i}\bar x_j\bigr],
\end{align*}
$G_N=\sum_{k=N}^{2N-1} \gamma_k^{(k)}$, $X=\sum_{k=1}^N x_k$ and
\begin{align*}
\mathrm C^{AB}_N(\gamma_1,\dots,\gamma_{2N-1}) =
\begin{cases}
1, & \text{odd } N,\\
(-1)^{\sum_{k=1}^{N}\bigl[\gamma_{2N-k}^{(k-1)}-\gamma_{N}^{(N-1)}\bigr]},
&
\text{even } N.
\end{cases}
\end{align*}
Similar to the previous case one can argue that $\Phi^{(N)}_{x}$ is a continuous function of $\nu_k$ in the vicinity of the real axis.

Finally, the scalar product of the functions $\Psi^{(N+1)}_{p,x}(z_1,\dots, z_{N+1})$ and $\Psi^{(N)}_{q_1,y}(z_1,\dots,z_N)\otimes
\Psi^{(1)}_{q_2}(z_{N+1})$ which we need in the proof of Theorem~\ref{theorem}, takes the form
\begin{gather}
\big(\Psi^{N}_{q_1,y}\otimes \Psi^{(1)}_{q_2},\Psi^{(N+1)}_{p,x}\big) =
\mathrm C_{NN+1}(\gamma) \pi\delta^{(2)}(p-q_1-q_2)
|p|^N |q_1|^{N-1}
\nonumber\\
\qquad\times
({\rm i}p)^{-\bar G_{N+1}^* }
(-{\rm i}\bar p)^{-G_{N+1}^* } ({\rm i} q_2)^{-\gamma'_{2N}}(-{\rm i}\bar q_2)^{-\bar \gamma'_{2N}}
(-{\rm i} q_1)^{-G_N } ({\rm i}\bar q_1)^{-\bar G_N }
\nonumber
\\
\qquad\times
\left(1+\frac {q_1} {q_2}\right)^{{\rm i}\bar Y^*}
\left(1+\frac{\bar q_1}{\bar q_2}\right)^{{\rm i}Y^*}
\left(-\frac {q_2} {q_1}\right)^{i X}
\left(-\frac{\bar q_2}{\bar q_1}\right)^{i \bar X}
\nonumber\\
\qquad\times \frac{\prod_{k=1}^{N-1}\prod_{j=1}^{N}\boldsymbol\Gamma\left[{\rm i}(\bar y_k^*-x_j)\right]}
{ \left(\prod_{j=1}^N\prod_{k=1}^{N-1} \boldsymbol\Gamma\bigl[\gamma_{2N-k}^{(k-1)}-{\rm i}x_j\bigr]\right)
\left(\prod_{k=1}^N\prod_{j=1}^{N-1} \boldsymbol\Gamma\bigl[\bar\gamma_{2N-k}^{(k-1)}-{\rm i}\bar y_j\bigr]\right)^\dagger },\label{bbN}
\end{gather}
where
\begin{align}\label{ANdef}
G_N=\sum_{m=N}^{2N-1}\gamma_m^{(m)}, \qquad G_{N+1}=G_N-\gamma_N^{(N)}=\sum_{m=N+1}^{2N-1}\gamma_m^{(m)}
\end{align}
and
\begin{align*}
\mathrm C_{NN+1}(\gamma_1,\dots,\gamma_{2N})=\begin{cases}
1, & \text{for odd}\ N,\\
 (-1)^{\sum_{k=1}^{N-1}\bigl[\gamma_{2N-k}^{(k-1)}-\gamma_N^{(N-1)}\bigr]}, & \text{for even}\ N.
\end{cases}
\end{align*}
The calculation is almost the same as in the previous cases so we omit the details.

\section{SoV representation}\label{sect:SoV}

In the previous section, we constructed the functions $\Psi_{p,x}^{(N)}$ and $\Phi^{(N)}_{x}$ associated with the entries~$B_N$ and $A_N$ of
the monodromy matrix~\eqref{monodromymatrix}. For a given vector $\Psi \in {\mathbb H}_N$, we define two functions by projecting it on
\smash{$\Psi_{p,x}^{(N)}$} and \smash{$\Phi^{(N)}_{x}$}:
\begin{align*}
\varphi(p,x_1,\dots,x_{N-1})=\big(\Psi_{p,x}^{(N)},\Psi\big), \qquad
\chi(x_1,\dots,x_N)=\big(\Phi_{x}^{(N)},\Psi\big).
\end{align*}
These functions are symmetric functions of the variables $x$. It was shown by Sklyanin~\cite{MR1239668} that the transformation
$\Psi\mapsto \varphi (\Psi\mapsto \chi)$ reduces the original multidimensional spectral problem for the transfer matrix to the set of
one-dimensional spectral problems that greatly simplifies the analysis. We want to show that the maps $\Psi\mapsto \varphi$ and
$\Psi\mapsto \chi$ can be extended to the isomorphism between the Hilbert spaces, ${\mathbb H}_N \mapsto \mathbb H_{{\rm SoV}}$.

Let us define
\begin{gather}
\left(\varphi_1,\varphi_2\right)_{B_N} =
\int_{\mathbb R\times \mathbb R} \int_{\mathscr D^\sigma_{N-1}} (\varphi_1(p,x))^\dagger \varphi_2(p,x) \mu_{N-1}\left(x\right)
{\rm d}^2p {\rm d}\mu^{B}_{N-1}(x),
\notag\\
\left(\chi_1,\chi_2\right)_{A_N} =
 \int_{\mathscr D^\sigma_{N}} (\chi_1(x))^\dagger \chi_2(x)
 {\rm d}\mu^{A}_{N}(x).\label{symscprod}
\end{gather}
The variables $x_k$, $\bar x_k$ take the form $x_k={\rm i}n_k/2+ \nu_k$, $\bar x_k=-{\rm i}n_k/2+ \nu_k$, where all $n_k$ are either integers or
half-integers,
\begin{align*}
n_k\in \mathbb Z^\sigma\equiv \mathbb Z+\frac\sigma 2, \qquad \sigma=0,1,
\end{align*}
and
\begin{align*}
\mathscr D^\sigma_N \equiv \left(\mathbb R \times Z^\sigma\right)^N.
\end{align*}
The measures are defined as follows:
\begin{align*}
 d\mu^{B(A)}_N(x)=\mu^{B(A)}_N(x) \prod_{k=1}^N\mathcal Dx_k, \qquad \mu^{B(A)}_{N}(x) =c_N^{B(A)} \mu_{N}(x).
\end{align*}
The symbol $\mathcal D x$ stands for
\begin{align*}
\int \mathcal D x \equiv \sum_{n\in \mathbb Z^\sigma} \int_{-\infty}^\infty {\rm d}\nu.
\end{align*}
The weight function $\mu_{N}(x)$ is given by the following expression:
\begin{align*}
\mu_{N}(x_1,\dots,x_N) & =\prod_{1\leq k<j\leq N} x_{kj}\bar x_{kj} =\prod_{1\leq k<j\leq N}\left(\nu_{kj}^2+\frac14 n_{kj}^2\right),
\end{align*}
where $x_{kj}=x_k-x_j$, $\nu_{kj}= \nu_{k}-\nu_j$, $n_{kj}=n_k-n_j$ while the coefficients $c_N^{B(A)}$ take the form
\begin{align*}
\left(c_{N}^B\right)^{-1}=\frac12{(2\pi)^{N+1} N!},\qquad \left(c_N^A\right)^{-1} ={(2\pi)^N N!}.
\end{align*}

Let ${\mathbb H}_N^{B,\sigma}$, ${\mathbb H}_N^{A,\sigma}$ be the Hilbert spaces of symmetric functions corresponding to the scalar
products~\eqref{symscprod}:
\begin{gather*}
{\mathbb H}_N^{B,\sigma} = L^2(\mathbb R\times\mathbb R)
 \otimes L^2_{\mathrm{sym}}\left( \mathscr D^\sigma_{N-1},
 {\rm d}\mu^{B}_{N-1}(x)\right),
\\
{\mathbb H}_N^{A,\sigma} =
 L^2_{\mathrm{sym}}\left( \mathscr D^\sigma_N,
 {\rm d}\mu^A_{N}(x)\right).
\end{gather*}

Given that $\varphi(p,x)$ and $\chi(x)$ are smooth and compactly supported functions on $\mathbb R^2\times \mathscr D^\sigma_{N-1} $
and $\mathscr D^\sigma_N$, respectively, we introduce transforms $\mathrm T_N^{B}\colon \varphi\mapsto \Psi_\varphi$ and $\mathrm
T_N^{A}\colon\chi\mapsto \Psi_\chi$,
\begin{subequations}
\begin{gather}
\label{TB}
\Psi_\varphi(z)\equiv \big[\mathrm T_N^{B}\varphi\big](z) =
\int_{\mathbb R^2}\int_{\mathscr D^{\sigma}_{N-1}}
\varphi(p,x)
\Psi_{p,x}^{(N)}(z) {\rm d}^2p {\rm d}\mu^{B}_{N-1}(x),
\\
\label{TA}
\Phi_\chi(z)\equiv \big[\mathrm T_N^{A}\chi\big](z) =
\int_{\mathscr D^{\sigma}_N} \chi(x)\Phi_{p,x}^{(N)}(z) {\rm d}\mu^{A}_{N}(x).
\end{gather}
\end{subequations}
Note that the function $\Psi_\varphi$ depends on the vector $\gamma$, equation~\eqref{gammaN}, which appears in the definition of the function
\smash{$\Psi_{p,x}^{(N)}$}. That is $\mathrm T_N^{B} \equiv \mathrm T_N^B(\gamma)$ and the same applies to the operator $\mathrm T_N^A$. In order
to not overload the notation, we do not display this dependence explicitly.

\subsection[B system]{$\boldsymbol{B}$ system}

We begin the proof of the unitarity of the transform $ \mathrm T_N^B $ with the following lemma.

\begin{Lemma}\label{firstlemma}
For any smooth fast decreasing function $\varphi$ on $\mathbb{R}^2\times\mathscr D^\sigma_{N-1}$, the function $ \mathrm T_{N}^{B}\varphi$
belongs to the Hilbert space ${\mathbb H}_N$ and it holds
\begin{align*}
\big\|\mathrm T_{N}^{B}\varphi\big\|_{{\mathbb H}_N}^2 =\|{\varphi}\|_{{\mathbb H}_N^{B,\pm}}^2 = \int_{\mathfrak D_N^\pm}
|\varphi(p, x)|^2 {\rm d}^2 p
 {\rm d}\mu_{N-1}^{B}( x).
\end{align*}
\end{Lemma}

\begin{proof} Let $\Psi_\varphi^\epsilon$ be a function defined by equation~\eqref{TB} with $\Psi_{p,x}^{(N)}$ replaced by
$\Psi_{p,x^\epsilon}^{(N),\epsilon}$, see equations~\eqref{Psiespilon} and \eqref{Psixepsilon}. It can be shown that
\smash{$\Psi_\varphi^\epsilon(\vec{p})\underset{\epsilon\to 0}{\mapsto} \Psi_\varphi(\vec{p})$} almost everywhere. Next, taking into account
equation~\eqref{BBproduct} one gets
\begin{align}
(\Psi_\varphi^\epsilon,\Psi_{\varphi'}^{\epsilon'})_{{\mathbb H}_N}
 ={}&\pi \int {\rm d}^2p \int {\rm d}\mu_{N-1}^{B}(x)\nonumber\\
 &\times\int {\rm d}\mu_{N-1}^{B}(x') (p\bar p)^{\epsilon+\epsilon'}
 \varphi(p,x)(\varphi^\prime(p,x'))^\dagger I^{\epsilon,\epsilon'}(x,x'),\label{PsiPsi}
\end{align}
with $ I^{\epsilon,\epsilon'}(x,x')$ given by equation~\eqref{twoform}.
 Let us assume that the function $\varphi(\varphi^\prime)$
has the form
\begin{align}\label{psifac}
\varphi(p,x_1,\dots,x_{N-1}) = \kappa(p) \phi(x_1,\dots,x_{N-1}),
\end{align}
where $\phi(x_1,\dots,x_{N-1})$ is a symmetric function
\begin{align}\label{ffact}
\phi(x_1,\dots,x_{N-1}) =
\sum_{S_{N-1}} \phi_1(x_{i_1})\cdots \phi_{N-1}(x_{i_{N-1}})
\end{align}
and the sum goes over all permutations. We also assume that the functions $\phi_k(x_k)=\phi_k(n_k,\nu_k)$ are local in $n_k$,
$\phi(n_k,\nu_k)=\delta_{n_k,m_k}\phi_k(\nu_k)$ and $\phi_k(\nu_k)$ is an analytic function of $\nu_k$ in some strip~$|\operatorname{Im}\nu_k|<\delta_k$ which vanishes sufficiently fast at $\nu_k\to \pm\infty$. Such functions form a dense subspace in the Hilbert
space ${\mathbb H}_N^{B,\sigma}$. Since the momentum integral in \eqref{PsiPsi} factorizes one has to consider the integrals over
 $x_k=(n_k, \nu_k)$,
$x'_k=(n'_k, \nu'_k)$, which have the form
\begin{gather}
 \int {\rm d}\mu_{N-1}^{B}(x)\int {\rm d}\mu_{N-1}^{B}(x') \cdots
\nonumber
 \\ \qquad{}\equiv \prod_{j=1}^{N-1}\sum_{n_j\in \mathbb Z+\frac\sigma2}\sum_{n'_j\in \mathbb Z+\frac\sigma2}
 \int_{-\infty}^\infty\cdots \int_{-\infty}^\infty
 \mu_{N-1}^{B}(\vec{n},\vec{\nu})\mu_{N-1}^{B}(\vec{n}',\vec{\nu}')\prod_{k=1}^{N-1}{\rm d}\nu_k {\rm d}\nu'_k \cdots.\label{SumInt}
\end{gather}
According to our assumptions, only finite number of terms contribute to the sum in~\eqref{SumInt}. Let us study behaviour of a particular
term in the sum in the limit $\epsilon,\epsilon'\mapsto 0$. The functions~$\phi$,~$\phi^\prime$ are smooth and fast decreasing functions
of $\nu$, $\nu'$. The function~$ I^{\epsilon,\epsilon'}(x,x')$ contains the factor~$\boldsymbol \Gamma[\epsilon+\epsilon' +{\rm i}\bar X-{\rm i}(X')^*]/
\boldsymbol \Gamma[\epsilon+\epsilon']$ and the product of the $\boldsymbol \Gamma$-functions
\begin{align}\label{afactor}
\boldsymbol\Gamma\left[{\rm i}((\bar x'_k)^*-x_j)\right] &=
\boldsymbol\Gamma\left[\frac{n'_k}2 - \frac{n_j}2 + {\rm i}(\nu'_k-\nu_j) + \epsilon_{jk}\right]\nonumber\\
&=\frac{\Gamma\bigl(\frac{n'_k}2 - \frac{n_j}2 + {\rm i}(\nu'_k-\nu_j) + \epsilon_{jk}\bigr)}{
\Gamma\bigl(1+\frac{n'_k}2 - \frac{n_j}2 - {\rm i}(\nu'_k-\nu_j) - \epsilon_{jk}\bigr)},
\end{align}
where $\epsilon_{jk}\equiv\epsilon_j+\epsilon'_k$. In the $\epsilon'_k, \epsilon_j\to 0$ this function becomes singular at $\nu'_k =\nu_j$
if $n'_k=n_j$. Let us shift the contours of integrations over $\nu'_k$ variables to the upper half-plane,
$\operatorname{Im}\nu'_k=\delta>\epsilon_{jk}$, and pick up the residues at the corresponding poles. After this, we can send $\epsilon'_k,
\epsilon_j\mapsto 0$. Let us consider a generic contribution arising after this rearrangement. It has the form
\begin{align*}
\int_{C_\delta} \cdots \int_{C_\delta}
 \prod_{k=1}^{M} {\rm d}\nu'_{i_k} f\left( x_1,\dots,x_{N-1}, S(x'_1),\dots, S(x'_{N-1})\right),
\end{align*}
where $ S(x'_k)= x'_k $ if $k \in (i_1,\dots, i_M)$ and $S(x'_k)= x_{p_k}$ if $k$ does not belong to this set. The integrand $f$ is given
by the product of the functions $\phi_k$, $\phi^\prime_k$, $\boldsymbol\Gamma$-functions~\eqref{afactor} and the factor~${A=\boldsymbol
\Gamma[\epsilon+\epsilon' +{\rm i}\bar X-{\rm i}(X')^*]/ \boldsymbol \Gamma[\epsilon+\epsilon']}$.
 All these factors
are regular on the contours of integration. Moreover, if $M\geq 1$ the last factor, $A$, tends to zero at $\epsilon,\epsilon'\mapsto 0$.
Thus the only non-vanishing contribution comes from the term with $M=0$, i.e., when all $x'_k\mapsto x_{i_k}$ for $k=1,\dots, N-1$. It
takes the form
\begin{align*}
(\Psi_\varphi^\epsilon,\Psi_{\varphi'}^{\epsilon'})_{{\mathbb H}_N}
 = \int {\rm d}^2p \int {\rm d}\mu_{N-1}^{B}(x)
 \varphi(p,x)(\varphi^\prime(p,x))^\dagger
 +O(\epsilon+\epsilon')
\end{align*}
that results in the following estimate for the norm of the function $\Psi^\epsilon_\varphi$:
\begin{align*}
\|\Psi^\epsilon_\varphi\|^2_{{\mathbb H}_N} = K + O(\epsilon),
\end{align*}
where
\begin{align*}
K =\|{\varphi}\|_{{\mathbb H}_N^{B,\sigma}}^2 \equiv \int_{\mathbb R^2}\int_{\mathscr D^{\sigma}_{N-1}}
|\varphi(p, x)|^2 {\rm d}^2 p
 {\rm d}\mu_{N-1}^{B}( x).
\end{align*}
Since $\Psi_\varphi^\epsilon(\vec{p})\mapsto \Psi_\varphi(\vec{p})$ at $\epsilon \to 0$, it follows from Fatou's theorem that
$\|\Psi_\varphi\|^2_{{\mathbb H}_N} < K$. At the same time, the inequality{\samepage
\begin{align*}
\|\Psi_\varphi-\Psi_\varphi^\epsilon \|^2_{{\mathbb H}_N} \geq 0
\end{align*}
implies $\|\Psi_\varphi\|^2_{{\mathbb H}_N}\geq K$ that results in $\|\Psi_\varphi\|^2_{{\mathbb H}_N}=K$.}

Since the set of functions~\eqref{psifac}, \eqref{ffact} is dense in the Hilbert spaces ${\mathbb H}_N^{B\pm}$, the transformation~$\mathrm
T_N^{B}$ can be extended to the entire Hilbert space ${\mathbb H}_N^{B,\pm}$ and equation~\eqref{unit-BA} holds for any function $\varphi \in \mathbb{H}_N^{B,\pm}$.
\end{proof}

Taking this result into account we formulate the following theorem.

\begin{Theorem}\label{theorem}
The map $\mathrm T_N^B$ defined in equation~\eqref{TB} can be extended to the linear bijective isometry of the Hilbert spaces, $\mathbb
H_N^{B,\sigma} \mapsto \mathbb H_N$, i.e.,
\begin{subequations}
\begin{align}\label{unit-BA}
\big\|\mathrm T_{N}^{B}\varphi\big\|_{{\mathbb H}_N}^2 &=\|{\varphi}\|_{{\mathbb H}_N^{B,\sigma}}^2
\\
\intertext{and}
\label{unit-Range}
\mathcal R\left( \mathrm T_{N}^{B}\right)
&= {\mathbb H}_N.
\end{align}
\end{subequations}
\end{Theorem}

\begin{proof} Equation~\eqref{unit-BA} is a direct consequence of Lemma~\ref{firstlemma}. It implies that $\big\|\mathrm T_{N}^{B}\big\|=1$,
hence $\mathcal R\big( \mathrm T_{N}^{B}\big)$ is a closed subspace in ${\mathbb H}_N$ and \smash{${\mathbb H}_N= \mathcal R\big( \mathrm
T_{N}^{B}\big)\oplus \mathcal R\big( \mathrm T_{N}^{B}\big)^\perp$}. Since \smash{$\mathcal R\big( \mathrm T_{N}^{B}\big)^\perp =
\ker\big(\mathrm T_N^B\big)^*$} in order to prove \eqref{unit-Range} it is enough to show that $ \ker\big(\mathrm T_N^B\big)^*=0$.
\end{proof}

We prove this statement using induction on $N$. For $N=1$, the map $\mathrm T_{N=1}^B$ is a two-dimensional Fourier transform, hence
equation~\eqref{unit-Range} is true. Let us now assume that $\mathcal R\left( \mathrm T_{N}^{B}\right) = {\mathbb H}_N$ and prove that it implies
$\mathcal R\left( \mathrm T_{N+1}^{B}\right)= {\mathbb H}_{N+1}$. As was stated above, it is sufficient to prove that \smash{$\ker\left(\mathrm
T_{N+1}^B\right)^*=0$}. To this end, let us consider the map
\begin{align*}
\mathrm S_N =\left(\mathrm T_{N+1}^B\right)^* \left(\mathrm T_{N}^B\otimes \mathrm T_{1}^B\right), \qquad
{\mathbb H}_{N}^{B,\sigma}\otimes L^2(\mathbb R^2)\overset{\mathrm T_{N}^B\otimes \mathrm T_{1}^B}{\longmapsto} {\mathbb H}_{N+1}
\overset{(\mathrm T_{N+1}^B)^*}{\longmapsto}{\mathbb H}_{N+1}^{B,\sigma}.
\end{align*}
Since by the assumption $\mathrm T_{N}^B\otimes \mathrm T_{1}^B$ is a bijective isometry $\ker \mathrm S_N\!=0$ if and only if ${\ker
\left(\mathrm T_{N+1}^B\right)^* \!= 0}$.

The adjoint operator $\big(\mathrm T_{N+1}^B\big)^*$ is a bounded operator which acts on a vector $\Psi\in {\mathbb H}_{N+1}$ by projecting it on the
eigenfunction \smash{$\Psi^{(N+1)}_{p,x}$},
\begin{align}\label{def-varphi}
\big(\mathrm T_{N+1}^B\big)^* \Psi = \big(\Psi^{(N+1)}_{p,x},\Psi\big)_{{\mathbb H}_{N+1}}=
\big(\Psi^{(N+1)}_{p,x},\mathrm P_{N+1}\Psi\big)_{{\mathbb H}_{N+1}}\equiv\varphi(p,x),
\end{align}
where $\mathrm P_{N+1}$ is the projector on $\mathcal R\big(\mathrm T^B_{N+1}\big)$. It follows from \eqref{def-varphi} that
\begin{align}\label{firstnormestimate}
\Vert \varphi\Vert_{{\mathbb H}_{N+1}^{B,\sigma}}^2 =\int_{\mathbb R^2}\int_{\mathscr D^{\sigma}_{N}}\vert \varphi(p,x)\vert^2 {\rm d}^2p {\rm d}\mu_N^B(x) =
\Vert \mathrm P_{N+1}\Psi\Vert^2_{{\mathbb H}_{N+1}} \leq \Vert \Psi\Vert^2_{{\mathbb H}_{N+1}}.
\end{align}
For $\phi\in {\mathbb H}_{N}^{B,\sigma}\otimes L^2\big(\mathbb R^2\big)$, the function $\Psi_\phi=\big(\mathrm T_{N}^B\otimes \mathrm
T_{1}^B\big)\phi$ reads
\begin{align}\label{TNT1}
 \Psi_\phi(z)
 = \int_{\mathbb{R}^2\otimes\mathbb{R}^2} \int_{\mathscr D^{\sigma}_{N-1}}
 \Psi^{(N)}_{q_1,x}(z_1,\dots,z_N) \Psi^{(1)}_{q_2}(z_{N+1})
\phi(q_1,q_2,x) {\rm d}^2q_1 {\rm d}^2q_2 {\rm d}\mu_{N-1}^B(x).
\end{align}
Replacing $ \Psi^{(N)}_{q_1,x}\mapsto \Psi^{(N),\epsilon}_{q_1,x}$ in \eqref{TNT1}, we define a new function, $\Psi_\phi^\epsilon$.
According to Lemma~\ref{firstlemma}, $\Psi_\phi^\epsilon\underset{\epsilon\to 0^+}{\longrightarrow} \Psi_\phi$ in ${\mathbb H}_{N+1}$ for smooth
rapidly decreasing functions, we obtain
\begin{align}
\varphi(p,x)&=[ S_N \phi](p,x)= \big(\Psi^{(N+1)}_{p,x},\Psi_\phi\big)_{{\mathbb H}_{N+1}}\nonumber\\
&=\lim_{\epsilon\to 0^+}
\big(\Psi^{(N+1)}_{p,x},\Psi^\epsilon_\phi\big)_{{\mathbb H}_{N+1}}\equiv \lim_{\epsilon\to 0^+}\varphi_\epsilon(p,x),\label{varphiSN}
\end{align}
where
\begin{align}\label{varphiSepsilon}
\varphi_\epsilon(p,x) =
\int_{\mathbb R^2\times \mathbb R^2 }\int_{\mathscr D^{\sigma}_{N-1}}
 S^\epsilon_N(p,x|q_1,q_2,x') \phi(q_1,q_2,x') {\rm d}^2q_2 {\rm d}^2q_1 {\rm d}\mu_{N-1}^B(x').
\end{align}
The kernel $S_N^\epsilon$ reads
\begin{align}\label{Sepsilon}
S^\epsilon_N(p,x|q_1,q_2,x') = \big(\Psi^{(N+1)}_{p,x},\Psi^{(N)}_{q_1,x'_\epsilon}\otimes \Psi^{(1)}_{q_2}\big),
\end{align}
see equation~\eqref{bbN}, and $x'_\epsilon = \left(x'_1+{\rm i}\epsilon_1,\dots,x'_{N-1}+{\rm i}\epsilon_{N-1}\right)$. We assume that function $\phi$
takes the form
\begin{align}\label{phidef}
\phi(q_1,q_2,x_1,\dots,x_{N-1}) = \kappa_1(q_1)\kappa_2(q_2) \sum_{S_{N-1}} \phi_1(x_{i_1})\cdots \phi_{N-1}(x_{i_{N-1}}),
\end{align}
where the sum goes over all permutations and that the functions $\phi_k$ are local in ``$n$" variable, that is
$\phi_k(x_k)=\phi_k(n_k,\nu_k)=\delta_{n_k m_k}\phi_{n_k}(\nu_k)$ and $\phi_{n_k}$ are compactly supported.
 The function~$\varphi(p,y)$ does not decrease sufficiently fast for large $y_k$ in order to justify changing the order of integration after
substituting $\varphi_\epsilon(p,y)$ in the form~\eqref{varphiSN}, \eqref{varphiSepsilon} into \eqref{firstnormestimate}. To overcome this
difficulty, we following the lines of~\cite{Derkachov:2021wja}, consider the integral
\begin{align*}
I_Z(\varphi) =
\int_{\mathbb R^2}\int_{\mathscr D^{\sigma}_{N}} |
\varphi(p,y)|^2 \Omega_Z(y) {\rm d}^2p {\rm d}\mu_{N}^B(y),
\end{align*}
where
\begin{align*}
 \Omega_Z(y)=\prod_{k=1}^{N} \frac{\boldsymbol\Gamma\left[ Z+{\rm i}y_k, Z-{\rm i}y_k\right]}
 {\boldsymbol\Gamma\left[Z,Z\right]}, \qquad Z=\bar Z =\frac12 + {\rm i} M.
\end{align*}
For $y_k^*=\bar y_k$ the factor $\Omega$ is a pure phase, $\vert\Omega_Z(y)\vert = 1$ and $\Omega_Z(y)\mapsto 1$ when $M\to \infty$, $y$
is fixed. Since the integral \eqref{firstnormestimate} is convergent,
\begin{align*}
\|\varphi\|^2_{{\mathbb H}_{N+1}^{B,\sigma}} =\lim_{M\to\infty}
\int_{\mathbb R^2}\int_{\mathscr D^{\sigma}_{N}} |
\varphi(p,y)|^2 \Omega_{Z}(y) {\rm d}^2p {\rm d}\mu_{N}^B(y).
\end{align*}

It follows from equations~\eqref{varphiSepsilon}, \eqref{Sepsilon} and \eqref{bbN} that for compactly supported functions $\phi_k$ the function
$f(\nu)=|\varphi_\epsilon(p,y)|^2$ is an analytic function of $\nu_k$ in the vicinity of the real axis for sufficiently large $\nu_k$.
Thus, we can write
\begin{align}\label{IomegaZ}
I_Z(\varphi) =\lim_{\omega\to 0} I^\omega_Z(\varphi) =
\lim_{\omega\to 0}\int_{\mathbb R^2}\int_{\mathscr D^{\sigma,\omega}_{N}} |
\varphi(p,y)|^2 \Omega_{Z-\omega}(y) {\rm d}^2p {\rm d}\mu_{N}^B(y),
\end{align}
where the integration contours over $\nu_k$ are deformed in order to separate the poles due to the Gamma functions, $\boldsymbol\Gamma
\left[ Z-\omega\pm {\rm i}y_k\right]$, in the factor $\Omega$. The integral $I^\omega_Z(\varphi)$ is an analytic function of~$\omega$.
Substituting $\varphi(p,y)$ in~\eqref{IomegaZ} in the form~\eqref{varphiSepsilon}, one can show that for $\text{Re}\omega>1$ the integrals
over $y$ decay fast enough to allow the change of the order of integration over $x$, $x'$ and~$y$. Thus, we obtain
\begin{align}
I^\omega_Z(\varphi) ={}&\lim_{\epsilon,\epsilon'\to 0^+} \int_{\mathbb R^2\times \mathbb R^2}
\int_{\mathscr D^{\sigma}_{N-1}\times \mathscr D^{\sigma}_{N-1}}
 \delta^{(2)}(q_1+q_2-q'_1-q'_2) \phi(q_1,q_2,x)
\left(\phi(q'_1,q'_2,x')\right)^\dagger
\nonumber\\
&\times\left|\frac{q'_1}{q_1}\right|^{N-1}
\left|\frac{q_1+q_2}{q_1 q'_2}\right|^{2}
\left(1+\frac{q'_1}{q'_2}\right)^{{\rm i}X'}\left(1+\frac{\bar q'_1}{\bar q'_2}\right)^{{\rm i}\bar X'}\left(1+\frac{q_1}{q_2}\right)^{-{\rm i}X}
\left(1+\frac{\bar q_1}{\bar q_2}\right)^{-{\rm i}\bar X}
\nonumber\\
&\times\left(\frac{q_1}{q'_1}\right)^{G_N}\left(\frac{\bar q_1}{\bar q'_1}\right)^{\bar G_N}
\left(\frac{q'_2}{q_2}\right)^{\gamma_{2N}}\left(\frac{\bar q'_2}{\bar q_2}\right)^{\bar \gamma_{2N}}
R(x,x')
J^{(\epsilon)}_\omega(Z,\zeta,x,x')\nonumber\\
&\times
{\rm d}^2q_1 {\rm d}^2q_2 {\rm d}^2q'_1 {\rm d}^2q'_2 {\rm d}\mu_{N-1}^B(x) {\rm d}\mu_{N-1}^B(x'),\label{Iomega}
\end{align}
where $\zeta = \dfrac{q_1 q'_2}{q_2 q'_1}$, $G_N$ is defined in equation~\eqref{ANdef},
\begin{align*}
R(x,x') =\prod_{k=1}^N\prod_{j=1}^{N-1}
{\boldsymbol\Gamma\left[\bar\gamma^{(k-1)}_{2N-k}-{\rm i} x'_j\right]} / {\boldsymbol\Gamma\left[\bar\gamma^{(k-1)}_{2N-k}-{\rm i} x_j\right]}
\end{align*}
and
\begin{align}
J^{(\epsilon,\epsilon')}_\omega(Z,\zeta,x,x') ={}&\pi^2\int_{\mathscr D^{\omega,\sigma}_{N}}
\zeta^{{\rm i} Y}\bar \zeta^{{\rm i}\bar Y}\prod_{j=1}^N
\frac{\boldsymbol \Gamma[Z-\omega \pm {\rm i} y_j]}{\boldsymbol\Gamma^2(Z)}\nonumber\\
&\times{\prod_{k=1}^{N-1}
\boldsymbol \Gamma[{\rm i}( \bar x_k - \bar y_j)]\boldsymbol \Gamma[{\rm i}( y_j - x'_k)]} {\rm d}\mu_{N}^B(y).\label{Jomega}
\end{align}
We recall that the variables $\nu_k$, $\nu'_k$, ($x_k= {\rm i}n_k/2+\nu_k$, $x'_k={\rm i}n'_k/2+\nu'_k$) have small negative (positive) imaginary parts,
$\operatorname{Im}\nu_k=-\epsilon_k$, $\operatorname{Im}\nu'_k=\epsilon'_k$, which must be send to zero at the end of the calculation.

 The integral~\eqref{Jomega} can be obtained in the closed form with the help of equation~\eqref{GustafsonII}. Indeed,
\begin{align*}
\prod_{1\leq j\neq k\leq N} \frac 1{\boldsymbol \Gamma[{\rm i}(y_k-y_j)]} & =\mu_N(y) (-1)^{\sum_{k<j} [{\rm i}(y_k-y_j)]}
\notag\\
\intertext{and}
\prod_{j=1}^N\prod_{k=1}^{N-1}\boldsymbol \Gamma[{\rm i}(\bar x_k -\bar y_j)] &=
\prod_{j=1}^N\prod_{k=1}^{N-1}\boldsymbol \Gamma[{\rm i}( x_k -y_j)] (-1)^{\sum_{j=1}^N\sum_{k=1}^{N-1} [{\rm i}(y_j-x_k)]},
\end{align*}
where $y_k = {\rm i}m_k/2 + \nu_k$, $\bar y_k = -{\rm i}m_k/2 + \nu_k$ and we recall that $[{\rm i}y_k]= {\rm i}(y_k-\bar y_k)=-m_k$. Taking into account that
\begin{align*}
(-1)^{\sum_{k<j} [{\rm i}(y_k-y_j)]} (-1)^{\sum_{j=1}^N\sum_{k=1}^{N-1} [{\rm i}(y_j-x_k)]} =(-1)^{\sum_{1\leq k<j\leq N-1} [{\rm i}(x_k-x_j)]},
\end{align*}
one finds that the integral \eqref{Jomega} is nothing else as Gustafson's integral~\eqref{GustafsonII} [$u_k\to {\rm i}y_k$ for all~$k$,
 $\{z_1,\dots,z_N\}\mapsto \{{\rm i}x_1,\dots,{\rm i} x_{N-1}, Z-\omega\}$ and $\{w_1,\dots,w_N\} \mapsto \{ -{\rm i}x'_1,\dots,-{\rm i}x'_{N-1}, Z-\omega
 \}$].
Thus, we obtain for $ J^{(\epsilon,\epsilon')}_\omega$,
\begin{align}
J^{(\epsilon,\epsilon')}_\omega(Z,\zeta,x,x') ={}& \pi (-1)^{\sum_{ k<j} [{\rm i}(x_k-x_j)]}
\frac{\zeta^{Z - \omega + {\rm i}X}}{(1 + \zeta)^{2(Z-\omega) + {\rm i}(X - X')}}
 \frac{\bar\zeta^{\bar Z - \omega + {\rm i}\bar X}}{(1 + \bar \zeta)^{2(\bar Z - \omega) + {\rm i}(\bar X - \bar X')}}
 \nonumber\\
 &\times
 \frac{\boldsymbol\Gamma[2Z-2\omega]}{\boldsymbol\Gamma^2[Z]}
\prod_{k=1}^{N-1}
\frac{\boldsymbol\Gamma\left[Z-\omega+{\rm i}x_k,Z-\omega-{\rm i}x_k'\right]}{{\boldsymbol\Gamma[Z,Z]}} \nonumber\\
 &\times
\prod_{k,j=1}^{N-1}\boldsymbol\Gamma[{\rm i}(x_k-x'_j)].\label{Jomega2}
\end{align}
Let us substitute this expression into~\eqref{Iomega} and calculate the corresponding limits. First of all, since all factors containing
$\omega$ are regular at $\omega,\epsilon_k,\epsilon'_k\to 0$ one can interchange the limits and first send $\omega\to 0$.

At $M\to\infty$ the integral over $q$, $q'$ is dominated by the contribution from the stationary point at $\zeta=1$,
\begin{align*}
\frac{\boldsymbol\Gamma[1+2{\rm i}M]}{\boldsymbol\Gamma^2\big[\frac12+{\rm i}M\big]}
\int {\rm d}^2\zeta \frac{\big(\zeta\bar\zeta\big)^{{\rm i}M+\frac12}}{ \big((1+\zeta)\big(1+\bar\zeta\big)\big)^{1+2{\rm i}M} }\varphi(\zeta) & \underset{M\to \infty}{=}
 \pi \varphi(1 )\left( 1 + O\left(\frac1{M^{1/2}}\right)\right).
\end{align*}
Taking this into account and expanding the first factor in the second line in~\eqref{Jomega2}, one gets for equation~\eqref{Iomega}
\begin{align}
I_{\omega=0}(Z) ={} & \lim_{\epsilon,\epsilon'\to 0^+}
\int_{\mathscr D^{\sigma}_{N-1}\times \mathscr D^{\sigma}_{N-1}}\!
 \phi(q_1,q_2,x)
\left(\phi(q_1,q_2,x')\right)^\dagger\pi(-1)^{\sum_{ k<j} [{\rm i}(x_k-x_j)]} {\rm i}^{N-N'}
\nonumber\\
&\times R(x,x') \left(1+\frac{q_1}{q_2}\right)^{{\rm i}(X'-X)}
\left(1+\frac{\bar q_1}{\bar q_2}\right)^{{\rm i}(\bar X'-\bar X)}\left(\frac M 2\right)^{2{\rm i}(\mathcal V-\mathcal V')}
\nonumber\\
&\times
\prod_{k,j=1}^{N-1}\boldsymbol\Gamma[{\rm i}(x_k-x'_j)+ \epsilon_{kj}]
{\rm d}^2q_1 {\rm d}^2q_2 {\rm d}\mu_{N-1}^B(x) {\rm d}\mu_{N-1}^B(x') + \cdots, \label{Iomega=0}
\end{align}
where ellipses stand for terms vanishing at $M\to\infty$ and
\begin{align*}
& x_k=\frac{ {\rm i}n_k}2+ \nu_k, \qquad x'_k= \frac{{\rm i}n'_k}2+ \nu'_k, \qquad \epsilon_{kj}=\epsilon_k+\epsilon'_j, \\
 & X=\sum_{k=1}^{N-1}x_k,
 \qquad \mathcal V=\sum_{k=1}^{N-1}\nu_k, \qquad N=\sum_{k=1}^{N-1}n_k,
\end{align*}
etc. The analysis of this integral is similar to the analysis of the integral~\eqref{PsiPsi}.\footnote{We do it assuming that the
functions $\phi_k(x_k)$ have the properties discussed around equation~\eqref{psifac}.}
 In the limit
$\epsilon,\epsilon'\to 0$ the poles of the Gamma functions, $x_k=x'_j$, approach the integration contour, while all other factors remain
regular. Let us shift the integration contour in $x_k$ to the upper complex half-plane picking up the residues at the poles at $x_k=x'_j$.
We recall that the Gamma functions develop poles only when $n_k=n'_j$, otherwise they are regular at $\nu_k=\nu'_j$. Afterwards, we can send
$\epsilon,\epsilon'\to 0$. The answer is given by the sum of terms
\begin{align*}
\idotsint
M^{{\rm i}\sum_{k=1}^m (\nu_{i_k}-\nu'_{j_k})}\times f_m(x,x'){\rm d}\nu_{i_1}\cdots {\rm d}\nu_{i_m} {\rm d}\nu'_1\cdots {\rm d}\nu'_{N-1},
\end{align*}
where $f_m(x,x')$ is a smooth function. Note, the contours of integration over $\nu$ variables lay in the upper half-plane, so that
\smash{$\vert M^{{\rm i}\sum_{k=1}^m (\nu_{i_k}-\nu'_{j_k})}\vert <1$} in the integration region. Since the functions~$f_m(x,x')$ are smooth functions
all such terms with $m>0$ vanish after integration in the limit~$M\to\infty$. Thus the only contribution with $m=0$, i.e., when
$x_k=x'_{k_j}$, survives in this limit. Then one obtains after some algebra
\begin{align*}
\Vert\varphi\Vert^2_{{\mathbb H}_{N+1}^{B,\sigma}} & =\Vert S_N\phi\Vert^2_{{\mathbb H}_{N+1}^{B,\sigma}}=
 \int_{\mathbb R^2\times \mathbb R^2}\int_{\mathscr D^{\sigma}_{N-1}}
\vert \phi(q_1,q_2,x)\vert^2 {\rm d}^2q_1 {\rm d}^2q_2 {\rm d}\mu_{N-1}^B(x)\\& = \|\phi\|^2_{{\mathbb H}_{N}^{B,\sigma}\otimes L^2(\mathbb R^2)}.
\end{align*}
Since the space of functions~\eqref{phidef} dense in ${\mathbb H}_{N}^{B,\sigma}\otimes L^2\big(\mathbb R^2\big)$ this relation can be extended to
the whole Hilbert space. Thus one concludes that $\ker S_N=0$, and, hence, $\ker\big(\mathrm T_{N+1}^B\big)^*=0$.

\subsection[A system]{$\boldsymbol{A}$ system}\label{sect:Asystem}

Using the results of the previous section it becomes quite easy to prove the unitarity of $\mathrm T_N^A$ transform. First, we prove an
analogue of the Lemma~\ref{firstlemma}.

\begin{Lemma}\label{secondlemma}
For any smooth fast decreasing function $\chi$ on $\mathscr{D}^\sigma_{N}$ the function $ \mathrm T_{N}^{A}\chi$, equation~\eqref{TA}, belongs
to the Hilbert space ${\mathbb H}_N$ and it holds
\begin{align}\label{unit-lemma-A}
\big\|\mathrm T_{N}^{A}\chi\big\|_{{\mathbb H}_N}^2 =\Vert{\chi}\Vert_{{\mathbb H}_N^{A,\sigma}}^2 = \int_{\mathscr D^{\sigma}_{N}}
|\chi(x)|^2
 {\rm d}\mu_{N}^{A}( x).
\end{align}
\end{Lemma}

\begin{proof} The proof is similar to the proof of the Lemma~\ref{firstlemma}. It suffices to prove~\eqref{unit-lemma-A} for
functions of the form
\begin{gather}\label{chidef}
\chi(x_1,\dots,x_{N}) = \sum_{S_{N}} \chi_1(x_{i_1})\cdots \chi_{N}(x_{i_{N}}),
 \qquad \chi_k(x_k)=\chi_k(n_k,\nu_k)=\delta_{n_km_k}\chi_k(\nu_k).
\end{gather}
We assume that the functions $\chi_k(\nu)$ are analytic in some strip near the real axis. Let us calculate the projection
\begin{align}\label{varphiPhi}
\varphi_\chi(p,y) & = \big( \Psi_{p,y}^{(N)}, \Phi_\chi \big) =\lim_{\epsilon\to 0}
 \int_{ \mathscr{D}^\sigma_N } \big( \Psi_{p,y}^{(N)},\Phi_{x+{\rm i}\epsilon}^{(N)}\big) \chi(x) {\rm d}\mu^A_{N}(x).
\end{align}
Here we have given the variables $x_k\to x_k+{\rm i}\epsilon_k$, $\epsilon_k=\bar\epsilon_k>0$ small imaginary parts which allows us to change
the order of integration. In order to show that
 $ \| \varphi_\chi\|_{\mathrm{H}_N^{B,\sigma}} = \| \chi \|_{\mathrm{H}_N^{A,\sigma} } $ we write
\begin{align*}
\Vert \varphi_\chi\Vert^2_{{\mathbb H}_{N}^{B,\sigma}} &= \int_{\mathbb R^2} \int_{\mathscr D^{\sigma}_{N-1}}
 \vert \varphi(p,y)\vert^2 {\rm d}^2p {\rm d}\mu_{N-1}^B(y)\\
 & =
 \lim_{\sigma\to 0} \int_{\mathbb R^2} {\rm e}^{-\sigma|p|^2}\bigg( \int_{\mathscr D^{\sigma}_{N-1}}\!
 \vert \varphi(p,y)\vert^2 {\rm d}\mu_{N-1}^B(y)\bigg) {\rm d}^2p.
\end{align*}
Using the representation~\eqref{varphiPhi} for $\varphi_\chi(p,y)$, we first evaluate the $y$-integral.\footnote{The $x$, $x'$, $y$ integral can
be interchanged since the integral of modulus is convergent.} This integral coincides with the so-called $\mathrm{SL}(2,\mathbb C)$
Gustafson integral and can be evaluated in a closed form~\eqref{GustafsonI} resulting in
\begin{align*}
\Vert \varphi_\chi\Vert^2_{{\mathbb H}_{N}^{B,\sigma}} ={}&\frac1\pi
 \lim_{\sigma\to 0} \lim_{\epsilon,\epsilon'\to 0^+}\int_{\mathbb R^2} \int_{\mathscr D^{\sigma}_{N}\times \mathscr
D^{\sigma}_{N}} {\rm e}^{-\sigma|p|^2} {\rm i}^{N-N'} p^{{\rm i}(X'-X)-1+\mathcal E+\mathcal E'}
\bar p^{{\rm i}(\bar X'- \bar X)-1+ \mathcal E+\mathcal E'}\\
&\times
 \chi(x) (\chi(x'))^\dagger
\frac{(-1)^{\sum_{k<j}[{\rm i}(x'_k -x'_j)]}}{\prod_{j=1}^{N}\vartheta_N(x_j)(\vartheta_N(x'_j))^\dagger}
\frac{\prod_{k,j=1}^N \boldsymbol\Gamma[{\rm i}(x'_k-x_j)+\epsilon_{jk}]}{\boldsymbol\Gamma[{\rm i}( X'- X )+\mathcal E+\mathcal E']}\\
&\times
 {\rm d}\mu^A_N(x) {\rm d}\mu^A_N(x') {\rm d}^2p,
\end{align*}
where $X=\sum_{k=1}^N x_k$, $N=\sum_{k=1}^N n_k$, $\mathcal E=\sum_{k=1}^N\epsilon_k$, $\epsilon_{jk}=\epsilon_j+\epsilon'_k$, etc. For the
momentum integral, one gets
\begin{align*}
\pi\delta_{NN'} \sigma^{{\rm i}(\mathcal V-\mathcal V') - \mathcal E -\mathcal E'} \Gamma({\rm i}(\mathcal V'-\mathcal V)+\mathcal E +\mathcal E'),
\end{align*}
where $\Gamma$ is Euler's gamma function. Thus
\begin{align*}
\Vert \varphi_\chi\Vert^2_{{\mathbb H}_N^{B,\sigma}} ={}&
\lim_{\sigma\to 0} \lim_{\epsilon,\epsilon'\to 0^+} \int_{\mathscr D^{\sigma}_{N}\times \mathscr
D^{\sigma}_{N}}
(-1)^{\sum_{k<j}[{\rm i}(x'_k -x'_j)]}\delta_{NN'} \sigma^{{\rm i}(\mathcal V-\mathcal V')}
\prod_{k,j=1}^N \boldsymbol\Gamma[{\rm i}(x'_k-x_j)+\epsilon_{jk}]\\
&\times\Gamma(1+{\rm i}(\mathcal V-\mathcal V'))
\frac{\chi(x)}{\Bigl(\prod_{j=1}^{N}\vartheta_N(x_j)\Bigr)}\!
\Bigg(\frac{\chi(x')}{\left(\prod_{j=1}^{N}\vartheta_N(x'_j)\right)}\Bigg)^\dagger\!
 {\rm d}\mu^A_N(x) {\rm d}\mu^A_N(x'),
\end{align*}
where we put $\epsilon_k,\epsilon'_k=0$ in all nonsingular factors. The analysis of this integral in the $\sigma, \epsilon,\epsilon'\to0$
limit is exactly the same as in Theorem~\ref{theorem}, see discussion around equation~\eqref{Iomega=0}, and results~in
\begin{align}\label{mapAnorm}
\Vert \Phi_\chi\Vert^2_{{\mathbb H}_N} & =
\Vert \varphi_\chi\Vert^2_{{\mathbb H}_N^{B,\sigma}} = \int_{\mathscr D^{\sigma}_{N}} \vert\varphi(x)\vert^2 {\rm d}\mu^A_N(x).
\end{align}
Since the space of the functions~\eqref{chidef} is dense in ${\mathbb H}_N^{A,\sigma}$, the relation~\eqref{mapAnorm} extends to the whole
Hilbert space.
\end{proof}

Finally, we formulate the analog of Theorem~\ref{theorem} for the map $\mathrm T_N^A$.
\begin{Theorem}\label{theoremA}
The map $\mathrm T_N^A$ defined in equation~\eqref{TA} can be extended to the linear bijective isometry of the Hilbert spaces, $\mathbb
H_N^{A,\sigma} \mapsto \mathbb H_N$, i.e.,
\[
\bigl\|\mathrm T_{N}^{A}\chi\bigr\|_{{\mathbb H}_N}^2 =\|{\varphi}\|_{{\mathbb H}_N^{A,\sigma}}^2
\]
and
\begin{align}\label{unit-RangeA}
\mathcal R\bigl( \mathrm T_{N}^{A}\bigr)
= {\mathbb H}_N.
\end{align}
\end{Theorem}

\begin{proof} As in the Theorem~\ref{theorem}, we only need to prove equation~\eqref{unit-RangeA}. As was discussed, earlier equation~\eqref{unit-RangeA} is equivalent to the statement that $\ker \big(\mathrm T_N^A\big)^*=0$ or to the assertion $\ker \mathbf S_N=0$,
where $\mathbf S_N= \big(\mathrm T_N^A\big)^* \mathrm T_N^B$. In order to prove this, it suffices to show that $\Vert \mathbf S_N
\varphi\Vert_{{\mathbb H}^{A,\sigma}_N} = \Vert \varphi\Vert_{{\mathbb H}^{B,\sigma}_N}$. The proof of this statement repeats step by step the
proof given in the Theorem~\ref{theorem}, and on the technical level is reduced to the evaluation of the integral~\eqref{Jomega}.
\end{proof}

\section{Summary}\label{sect:summary}

In this work, we consider a generic inhomogeneous $\mathrm{SL}(2,\mathbb C)$ spin chain with impurities and construct the eigenfunctions of
the $B$ and $A$ entries of the monodromy matrix. We prove the unitarity of the SoV transform associated with these systems or,
equivalently, the completeness of the corresponding systems in the Hilbert space of the model. Namely, the following identities hold in the
sense of distributions:
\begin{gather*}
\int_{\mathbb R^2} \int_{\mathscr D^\sigma_{N-1}} \Psi^{(N)}_{p,x}(z)\big(\Psi^{(N)}_{p,x}(z')\big)^\dagger
{\rm d}^2p {\rm d}\mu_N^B(x) =\prod_{k=1}^N \delta^2(z_k-z'_k),\\
\int_{\mathscr D^\sigma_N} \Phi^{(N)}_x(z)\big(\Phi^{(N)}_{x}(z')\big)^\dagger {\rm d}\mu_N^A(x) =\prod_{k=1}^N \delta^2(z_k-z'_k),
\end{gather*}
and
\begin{gather*}
\int_{\mathbb C^N} \Psi^{(N)}_{p,x}(z) \big(\Psi^{(N)}_{p',x'}(z)\big)^\dagger \prod_{k=1}^N {\rm d}^2z_k =
 \big(\mu^B_N(x)\big)^{-1}\delta^{2}(p-p') \delta^{N-1}(x, x'),\\
\int_{\mathbb C^N} \Phi^{(N)}_x(z) \big(\Phi^{(N)}_{x'}(z)\big)^\dagger \prod_{k=1}^N {\rm d}^2z_k =
 \big(\mu^A_N(x)\big)^{-1} \delta^N(x, x'),
\end{gather*}
where
\begin{align*}
\delta^N(x, x')=\frac1{N!}\sum_{w\in S_N} \delta^N(x'-wx), \qquad w x=(x_{w_1},\dots, x_{w_N} )
\end{align*}
and \[\delta^N(x'-x)=\prod_{k=1}^N \delta^2(x'_k-x_k), \qquad \delta^2(x'-x)=\delta_{nn'}\delta(\nu-\nu').\]

The method relies heavily on the use of multidimensional Mellin--Barnes integrals which generalize integrals calculated by
R.A.~Gustafson~\cite{Gustafson94}. The attractive feature of our approach is that it does not depends on the details of the spin chain
such as spins and inhomogeneity parameters. We believe that this technique can also be used to prove the unitarity of the SoV transform for
the open $\mathrm {SL} (2,\mathbb C)$ spin chain.

\appendix

\section{The diagram technique}\label{sect:Diagram}

Throughout this paper, we used a diagrammatic representation for the functions under consideration. The calculation of relevant scalar
products is, most conveniently, performed diagrammatically with the help of a few simple identities. Below, we give some of these rules
(see also~\cite{Derkachov:2001yn}).
\begin{enumerate}\itemsep=0pt
\item[(i)] An arrow with the index $\alpha$ directed from $w$ to~$z$ stands for a propagator $D_\alpha(z-w)=[z-w]^{-\alpha}$:

\centerline{\includegraphics[width=0.38\linewidth]{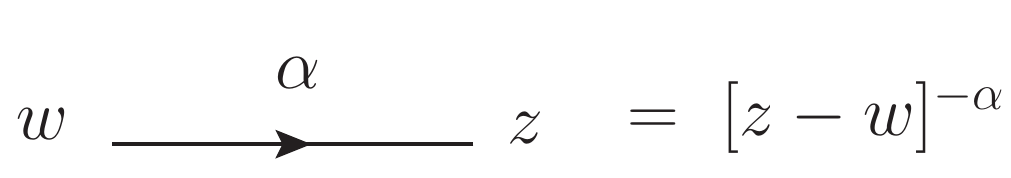}}

\item[(ii)] The Fourier transform reads
\begin{align*}
\int {\rm d}^2 z {\rm e}^{{\rm i}(pz+\bar p\bar z)}D_\alpha(z)=\pi {{\rm i}^{\alpha-\bar\alpha}}a(\alpha) D_{1-\alpha}(p),
\end{align*}
where the function $a(\alpha)\equiv 1/\boldsymbol\Gamma[\alpha]=\Gamma(1-\bar\alpha)/\Gamma(\alpha)$. 

\item[(iii)] Chain rule
\begin{align*}
\int\frac{ {\rm d}^2 w}{[z_1-w]^\alpha [w-z_2]^{\beta}}=
\pi
\frac{a(\alpha,\beta)}{a(\gamma)}
\frac{1}{[z_1-z_2]^{\gamma}},
\end{align*}
where $\gamma=\alpha+\beta-1$. Its diagrammatic form is

\centerline{\includegraphics[width=0.42\linewidth]{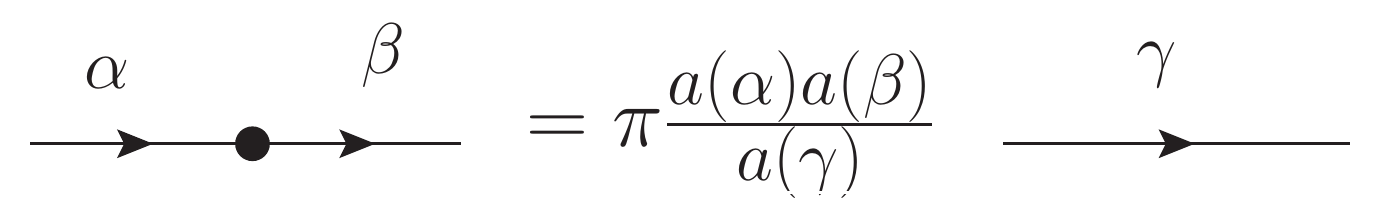}}

\item[(iv)] Star-triangle relation
\begin{equation*}
\vcenter{\hbox{
\includegraphics[width=0.57\linewidth]{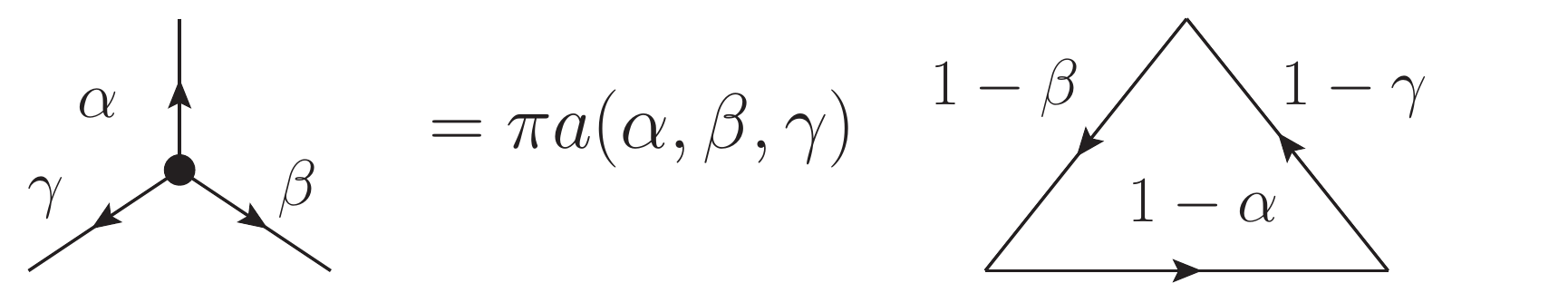}
}}
\end{equation*}
\item[(v)] Exchange relation
\begin{equation}
\label{exchange-rel}
\vcenter{\hbox{\includegraphics[width=0.57\linewidth]{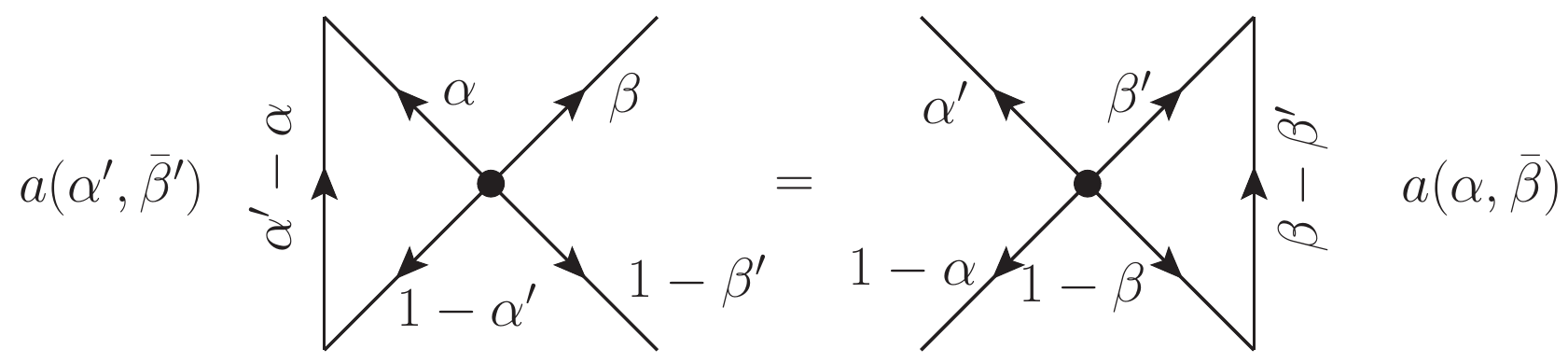}}},
\end{equation}
where $\alpha+\beta=\alpha'+\beta'$.
\end{enumerate}

\section{Scalar products}\label{app:scalarproduct}

Here, we discuss the calculation of scalar products of $\Psi^{(N),\epsilon}_{p,x}$ and $\Phi_{x}^{(N)}$ functions. The diagrams for the
scalar products~\eqref{BBproduct}, \eqref{ABproduct} are shown in Figure~\ref{diag:scalarproducts}. The leftmost vertex on both diagrams has
only two propagators attached to it. We call such a vertex -- free vertex. On the first step one integrates over the free vertex (on both
diagrams) using the chain relation for propagators and move the resulting line to the right with the help of the exchange relation. After
that two new free vertices appear and one repeat the same procedure again. In this way one can integrate over all vertices on the left
edge of both diagrams (they are shown by black blobs). Keeping trace of all factors arising in the process, one represent the initial
diagram $D$ as
\begin{gather}
D_N\big(\{x_1,x_2,\dots\}, \{y_1, y_2, \dots\}, \{\gamma_1, \gamma_2,\dots\} \big) \nonumber\\
\qquad {} = f(x_1,y_1,\gamma) D_N^\prime \big (\{x_2,\dots\}, \{y_2,\dots\}, \{\gamma_3,\dots\} \big ).\label{firststep}
\end{gather}
Taking into account that the function $\Psi_{p,x}^{(N)}$ and $\Phi_x^{(N)}$ are symmetric functions of the separated variables it follows
from \eqref{firststep} that
\begin{align}\label{DCN}
D_N\big(\{x_1,x_2,\dots\}, \{y_1, y_2, \dots\}, \{\gamma_1, \gamma_2,\dots\} \big) = \mathscr C_N(\gamma) \prod_{k,j} f(x_k,y_j,\gamma).
\end{align}
The factor $\mathscr C_N(\gamma)$ does not depend on $x$, $y$ variables. The easiest way to fix it is to evaluate both sides of \eqref{DCN}
for special values of $x$, $y$. For example, one can take $x_k\to x$ and $y_k\to \bar x^*$. Both sides, in this limits, contain divergent
factors, $\boldsymbol\Gamma[{\rm i}(\bar y^*_j- x_k)]$ which cancel out. It is easy check that the result of the integration over any free vertex
in this limit (after removing this singular factor) gives one. Therefore, the equation on $\mathscr C_N(\gamma)$ for the scalar
product~\eqref{twoform} takes the form
\begin{align*}
1=\mathscr C_N(\gamma) (\chi( x) (\bar \chi(\bar x^*))^*)^{N-1}=\mathscr C_N(\gamma)
(-1)^{(N-1)\sum_{k=0}^{N-3}\bigl[\gamma_{2N-3-k}^{(k)}-{\rm i}x\bigr]}.
\end{align*}
Since $\big[\gamma_m^{(k)}-{\rm i}x\big]$ is an integer number, one gets that $\mathscr C_N=1$ for odd $N$, while for even $N$
\begin{align*}
\sum_{k=0}^{N-3}\big[\gamma_{2N-3-k}^{(k)}-{\rm i}x\big] & =\sum_{k=0}^{N-3}\big(\big[\gamma_{2N-3-k}^{(k)} -\gamma_N^{(N-3)}\big]+\big[\gamma_N^{(N-3)}-{\rm i}x\big]\big)
\\
&=\sum_{k=0}^{N-3}\big[\gamma_{2N-3-k}^{(k)} -\gamma_N^{(N-3)}\big] +(N-2)\big[\gamma_N^{(N-3)} - {\rm i}x\big].
\end{align*}
Taking into account that the last term in the above equation is an even number, one gets that~$\mathscr C_N(\gamma)$ is given by the
expression~\eqref{BBCN}. For the second diagram, the analysis follows exactly the same lines.

\begin{figure}[t]
\centering
\includegraphics[width=0.8\linewidth]{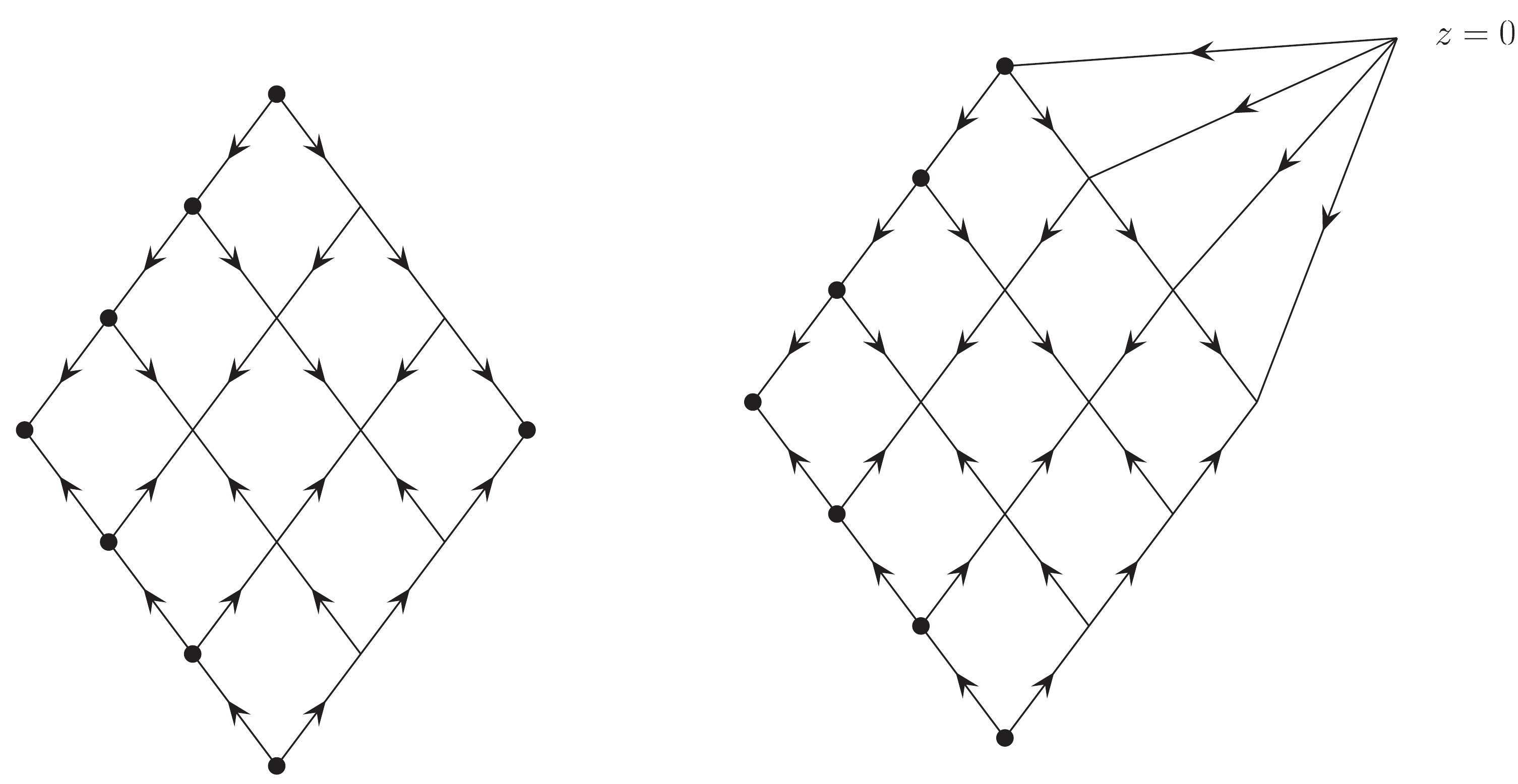}

\caption{Examples of diagrams for scalar products,
equations~\eqref{BBproduct}, \eqref{ABproduct} for $N=4$.}\label{diag:scalarproducts}
\end{figure}

\section{Gustafson's integral reduction}\label{app:gustafson}
The extension of the first Gustafson integral~\cite[Theorem 5.1]{Gustafson94} to the complex case was obtained in~\cite{Derkachov:2019ynh}.
It takes the form
\begin{gather}
\prod_{j=1}^N\sum_{n_j\in\mathbb{Z}+\frac{\sigma}{2}}
\int_{-{\rm i}\infty}^{{\rm i}\infty}
\frac{\prod_{m=1}^{N+1}\prod_{k=1}^{N}
\boldsymbol{\Gamma}(z_m-u_k)\boldsymbol{\Gamma}(u_k+w_m)}{ \prod_{m<j}
\boldsymbol{\Gamma}(u_m- u_j)\boldsymbol{\Gamma}(u_j- u_m)}
\prod_{p=1}^{N}
\frac{{\rm d}\nu_p}{2\pi {\rm i}}\nonumber\\
\qquad=\frac{N!\prod_{k,j=1}^{N+1}\boldsymbol{\Gamma}(z_k+ w_j)}{
\boldsymbol{\Gamma}\left(\sum_{k=1}^{N+1} (z_k+w_k) \right)},\label{GustafsonI}
\end{gather}
where $\boldsymbol \Gamma$ is the Gamma function of the complex field $\mathbb C$~\cite{MR2125927}
\[
\boldsymbol\Gamma(u) \equiv \boldsymbol\Gamma(u,\bar u) = \frac{\Gamma(u)}{\Gamma(1-\bar u)}=\frac1{a(u)}.
\]
The variables $u_k$, $w_m$, $z_m$ have the form
\begin{alignat*}{4}
&u_k=\frac{n_k}2+ \nu_k, \qquad &&z_m=\frac{n_m}2+ x_m, \qquad &&w_m=\frac{\ell_m}2+ y_m,&\\
&\bar u_k=-\frac{n_k}2+ \nu_k, \qquad &&\bar z_m=-\frac{n_m}2+ x_m, \qquad &&\bar w_m=-\frac{\ell_m}2+ y_m.&
\end{alignat*}
and the integration contours over $\nu_k$ separate the series of poles associated with the $\boldsymbol \Gamma$-functions:
$\boldsymbol\Gamma(z_m-u_k)$ and $\boldsymbol{\Gamma}(u_k+w_m)$, see~\cite{Derkachov:2019ynh} for more detail. The integral converges
for \[\sum_{m=1}^{N+1} \mathrm{Re}(z_m+w_m) <1.\]

Let us put
\begin{alignat*}{3}
&z_{N+1} =M\left(\frac12+{\rm i} x\right), \qquad &&\bar z_{N+1} =M\left(-\frac12+{\rm i} x\right),&\\
&w_{N+1} =M'\left(\frac12+{\rm i}x'\right), \qquad &&\bar w_{N+1} =M'\left(-\frac12+{\rm i}x'\right)&
\end{alignat*}
and send $M,M'\to\infty$ keeping $M/M'=\xi$ fixed, so that $w_{N+1}/z_{N+1}\mapsto \zeta$ and $\bar w_{N+1}/\bar z_{N+1}\mapsto \bar
\zeta$.

Dividing both sides of~\eqref{GustafsonI} by $(\boldsymbol \Gamma(z_{N+1})\boldsymbol \Gamma(w_{N+1}))^N$, we get in this limit
\begin{gather}
\frac1{N!}\prod_{j=1}^N\sum_{n_j\in\mathbb{Z}+\frac{\sigma}{2}}
\int_{-{\rm i}\infty}^{{\rm i}\infty} [\zeta]^U
\frac{\prod_{m,k=1}^{N}
\boldsymbol{\Gamma}(z_m-u_k)\boldsymbol{\Gamma}(u_k+w_m)}{ \prod_{m<j}
\boldsymbol{\Gamma}(u_m- u_j)\boldsymbol{\Gamma}(u_j- u_m)}
\prod_{p=1}^{N}\frac{{\rm d}\nu_p}{2\pi {\rm i}}\nonumber\\
\qquad=\frac{[\zeta]^Z}{[1+\zeta]^{Z+W}}\prod_{k,j=1}^{N}\boldsymbol{\Gamma}(z_k+ w_j),\label{GustafsonII}
\end{gather}
where $|\arg \zeta|<\pi$, $Z=\sum_{k=1}^N z_k$, $W=\sum_{k=1}^N w_k$ and we recall that $[\zeta]^U\equiv \zeta^U \bar\zeta^{\bar U}$.

\subsection*{Acknowledgements}

The author is grateful to S.\'{E}.~Derkachov for fruitful discussions and T.A.~Sinkevich for critical remarks.

\pdfbookmark[1]{References}{ref}
\LastPageEnding

\end{document}